\documentclass{article}

\PassOptionsToPackage{numbers,sort&compress}{natbib}

 \usepackage[preprint]{neurips_2026}

\usepackage[utf8]{inputenc} 
\usepackage[T1]{fontenc}    
\usepackage{hyperref}       
\usepackage{url}            
\usepackage{booktabs}       
\usepackage{amsfonts}       
\usepackage{nicefrac}       
\usepackage{microtype}      
\usepackage{xcolor}         
\usepackage{amssymb}
\usepackage{mathtools}
\usepackage{amsthm}
\usepackage{subfigure}
\usepackage{multirow}
\usepackage{caption}
\usepackage{graphicx}
\usepackage{amsmath}
\usepackage{cleveref}
\usepackage[table]{xcolor}
\usepackage{makecell}
\usepackage{dsfont}
\usepackage{enumitem}
\usepackage{placeins}

\definecolor{pastelyellow}{RGB}{255, 246, 173}
\definecolor{pastelorange}{RGB}{255, 212, 130}
\definecolor{pastelred}{RGB}{255, 179, 154}

\DeclareMathOperator{\diag}{diag}
\DeclareMathOperator{\rank}{rank}

\newcommand{\ReLU}{\operatorname{ReLU}}
\newcommand{\diam}{\operatorname{diam}}

\newtheorem{theorem}{Theorem}

\newtheorem{corollary}{Corollary}
\newtheorem{remark}{Remark}

\newtheorem{gridlemma}{Lemma}

\newtheorem{ffnlemma}{Lemma}

\newtheorem{ingplemma}{Lemma}

\newtheorem{gaplemma}{Lemma}

\title{Bounding Global and Local Compression Error of Signal Parameterizations}

\author{%
  Quang Luong Nhat Nguyen \qquad
  Sara Fridovich-Keil\\
  School of Electrical and Computer Engineering\\
  Georgia Institute of Technology\\
  \texttt{\{qnguyen95, sfk\}@gatech.edu} }

\begin{document}

\maketitle

\begin{abstract}
  Compressive signal parameterizations such as implicit neural representations and hybrid models are increasingly central to computational imaging, yet principled tools for evaluating reconstruction fidelity at finite model size remain limited when ground truth is unavailable. We introduce a framework for predicting the reconstruction error of compressive signal parameterizations, yielding non-asymptotic, signal-specific bounds that are both theoretically sound and efficiently computable without access to the ground truth signal. Specifically, we prove that when parameterization-based compression satisfies certain natural properties, the compression error at any compression level is bounded by a simple scaled difference between model predictions at different compression levels. We verify these properties for representative model families including interpolated grids, Fourier feature networks, multi-resolution hash encodings, and tensor factorizations, and show empirically that the resulting worst-case guarantees can be efficiently adapted into signal-specific error predictors that are tighter and generalizable. Across direct fitting of synthetic and natural signals, and inverse problems including radiance field and MRI reconstruction, our method closely tracks global error curves and yields informative local error heatmaps without ground-truth access. 
  Code is available at \url{https://github.com/voilalab/global_error_bound}.
\end{abstract}

\vspace{-1mm}
\section{Introduction} \label{sec:intro}
Computational imaging seeks to recover high-fidelity signals from indirect, noisy measurements, and modern applications increasingly demand higher resolution, higher dimensionality, and dynamic content, ranging from gigavoxel volumetric MRI~\citep{ong2020extrememri} to spatiotemporal neural radiance fields~\citep{jiao2024ffeinr, pumarola2021dnerf}. These applications introduce severe bottlenecks in memory and compute. Compressive signal representations have therefore become indispensable, and a growing body of work has explored compact parameterizations for computational imaging, including sparse grids~\citep{fridovich2022plenoxels}, basis decompositions~\citep{bartels1995introduction, broughton2018discrete}, coordinate-based neural networks~\citep{chen2021learning, lindell2021bacon, saragadam2023wire, sitzmann2020siren, tancik2020ffn, strumpler2022implicit}, multi-resolution hash encodings~\citep{muller2022ingp}, discrete mesh or point-based models~\citep{kerbl3Dgaussians, siddiqui2023meshgpt, yin2023continuouspde}, low-rank factorizations~\citep{chen2022tensorf,fridovichkeil2023kplanes,sivgin2024gaplanes}, and other hybrid models~\citep{chen2023dictfields, lombardi2021mvp} that combine several of these strategies. 

Despite their empirical success, these models offer limited practical guidance about what is lost under compression at finite model size. Moreover, their fidelity is often assessed by empirical comparison against a known ground truth signal~\citep{strumpler2022implicit}, which is not usually available in computational imaging applications.
In fields such as biomedical imaging and robotic perception, quantifying reconstruction fidelity in the absence of ground truth is essential for clinical diagnostics~\citep{ong2020extrememri, heckel2024deeplearningmri,lustig2007compressedsensing, molaei2023inrmedicalimaging} and operational safety~\citep{wang2025nerfroboticssurvey, araujo2024robotoicperceptiohn, zhu2022niceslam, sucar2021imap}. Similar needs arise in remote sensing and edge or bandwidth-limited imaging systems~\citep{gunasheela2018satellite, gomes2025geospatialcompression, louys1999astronomical, offringa2016compression, ngo2025dnnedge, yin2025surveyedge, delprete2025optimizing}, where communication, memory, and compute are tightly constrained.

While compression loss, or rate--distortion, is a well-studied problem in information theory~\citep{cover2006elements, shannon1959probability}, classical results and the neural compression pipelines built on them~\citep{strumpler2022implicit, chen2021nerv, dai2025implicit, guo2023bayesianinrs} measure the size of a compressed representation in bits. The signal parameterizations now common in computational imaging are instead collections of $d$ real-valued parameters, so bit-level results do not directly bound their error at a given $d$. Where theoretical approximation error bounds do exist for neural networks, they often remain confined to asymptotic or overparametrized regimes~\citep{achour2022nnlowerbound, siegel2020nnupperbound}, providing guarantees that are too coarse to inform signal-specific compression in practice. Consequently, there exist no principled tools to determine, without exhaustive search or access to ground truth, whether a given compressive real-valued parameterization is sufficiently accurate for a target application.

To address this gap, we derive finite-parameter error bounds for widely used differentiable parameterizations, including Fourier Feature Networks (FFN~\citep{tancik2020ffn}), multi-resolution hash encodings (Instant-NGP~\citep{muller2022ingp}), and tensor-factorized feature grids (Geometric Algebra Planes~\citep{sivgin2024gaplanes}), by establishing three structural properties of their optimal error curves: monotonicity, convexity, and zero error beyond a critical parameter budget. We further complement these guarantees with a simple prediction scheme that requires fitting only a few models at nearby compression levels, translating conservative worst-case bounds into efficiently computable and empirically tight error estimates. We validate this framework across two settings: direct signal fitting, spanning synthetic signals and natural images, and imaging inverse problems, including radiance field and accelerated MRI reconstruction. 

Concretely, we make the following contributions:
\begin{itemize}[leftmargin=15pt]
    \item We propose a framework for bounding the compression error of any signal parameterization that satisfies three natural properties: the compression error decreases monotonically with model size, is convex in model size, and is zero above some finite critical parameter budget. Under these assumptions, the compression error of a model can be rigorously bounded using only comparisons between models with nearby parameter counts, without any knowledge of the ground truth signal.
    \item We prove that interpolated grids as well as three representative and popular parameterizations from the computer vision literature (Fourier feature networks~\citep{tancik2020ffn}, multi-resolution hash encodings~\citep{muller2022ingp}, and tensor factorizations~\citep{sivgin2024gaplanes}) satisfy these properties under explicit assumptions, using squared-error or convex envelope formulations where needed.
    \item Since the worst-case theoretical bounds are provably valid but not empirically tight, we introduce a simple method to convert them into practically useful error estimates and validate these across all four representations and a variety of data types and inverse problems, including fitting natural images, synthetic images and volumes, 3D objects, and 3D radiance field and both 2D and 3D MRI reconstruction.
    \item Although our formal theory is stated for global reconstruction error, we demonstrate empirically that the same approach also yields informative local pixel-level error estimates that align well with the true error map, enabling efficient visualization and localization of compression artifacts without access to ground truth. 
\end{itemize}

\vspace{-1mm}
\section{Related Work} \label{sec:relworks}
\vspace{-1mm}
\paragraph{Compressive parameterizations.}
Compact signal parameterizations have found broad application in computational imaging tasks where memory and compute are tightly constrained, including accelerated MRI~\citep{ong2020extrememri, heckel2024deeplearningmri,lustig2007compressedsensing, molaei2023inrmedicalimaging}, robotic perception~\citep{wang2025nerfroboticssurvey, araujo2024robotoicperceptiohn, zhu2022niceslam, sucar2021imap}, remote sensing~\citep{gunasheela2018satellite, gomes2025geospatialcompression, louys1999astronomical, offringa2016compression}, and edge-constrained imaging~\citep{ngo2025dnnedge, yin2025surveyedge, delprete2025optimizing}. 
Recent architectures for these tasks have converged around a few families of differentiable parameterizations, including coordinate-based neural networks~\citep{chen2021learning, lindell2021bacon, saragadam2023wire, sitzmann2020siren, tancik2020ffn, strumpler2022implicit}, trainable spatial encodings~\citep{muller2022ingp}, and tensor-factorized models~\citep{chen2022tensorf,fridovichkeil2023kplanes,sivgin2024gaplanes}. In this work we focus on Fourier Feature Networks (FFNs)~\citep{tancik2020ffn}, Instant-NGP~\citep{muller2022ingp}, and Geometric Algebra Planes (GA-Planes)~\citep{sivgin2024gaplanes} as representative models of each type. All three models contain an MLP component, which we restrict (for the sake of analysis) to be a shallow single-hidden-layer MLP with finite width. A more detailed discussion of each model architecture is provided in Appendix~\ref{appendix:relworks:1}. 

\vspace{-1mm}
\paragraph{Rate-distortion theory, compression bounds, and approximation theory.}
Prior theory addresses neighboring but ultimately different questions. Classical rate-distortion theory~\citep{shannon1959probability, cover2006elements} and neural compression pipelines based on quantization or entropy coding~\citep{strumpler2022implicit, chen2021nerv, dai2025implicit, guo2023bayesianinrs} study compression under bit budgets through distributional distortion limits or explicit coding schemes, while complementary research~\citep{han2025losslessinr} studies the model size required for lossless neural representation of discrete signals at a given bit precision. We instead consider representations with $d$ real-valued parameters and bound the error for a single signal as a function of $d$ without assuming a source distribution or specifying a coding scheme. Separately, approximation theory~\citep{barron1993universal, siegel2020nnupperbound, achour2022nnlowerbound} and universal approximation results~\citep{ hornik1989multilayer, park2021minimumwidth, kim2024tighterminimumwidth} for neural networks characterize expressivity, approximation rates, and the existence of exact representations as width grows. These works provide important motivation and technical ingredients, but they do not yield non-asymptotic, signal-specific error bounds for compressive parameterizations as a function of model size that are both theoretically justified and computable without access to ground truth. Our work instead studies the geometry of the optimal error curve and shows how adjacent-model differences can be converted into non-asymptotic error bounds and practically useful predictors of compression error. We provide an expanded discussion of rate-distortion theory, approximation rates, and universal approximation results in Appendices~\ref{appendix:relworks:2}--\ref{appendix:relworks:4}.

\vspace{-1mm}
\section{Theoretical Results} \label{sec:theoreticalres}
We first formalize the properties that allow us to translate model differences between neighboring compression levels into rigorous error bounds. Let $x \in \mathbb{R}^n$ denote the true signal we aim to represent, such as an image or volume, and let $g_d: \mathbb{R}^d \to \mathbb{R}^n$ be a chosen compressive parameterization at compression level $d/n$, with learnable parameters $\theta \in \mathbb{R}^d$. For each parameter budget $d$, we define the optimal compression error as 
\begin{equation}\label{eq:gap_definition}
\|x - g_d^\star\| \coloneq \min_{\theta \in \mathbb{R}^d} \|x - g_d(\theta)\|,
\end{equation}
where $g_d^\star$ denotes the optimal size-$d$ representation of $x$ under parameterization architecture $g$ and $\|\cdot\|$ denotes the $\ell_2$ norm. Directly evaluating this quantity is generally intractable: the ground-truth signal $x$ is unobserved for many computational imaging applications, and even when $x$ is accessible the minimization over $\theta$ is nonconvex for neural or hybrid models $g$. This motivates our goal of bounding $\|x - g_d^\star\|$ using only the observable difference between trained models $g_{d_1}$ and $g_{d_2}$ at adjacent compression levels. For notational simplicity, we state the theory using consecutive parameter budgets $d$ and $d+1$; in practice, the parameter count of most architectures cannot be increased by exactly one, so we instead compare neighboring feasible budgets $d < d'$. The same telescoping argument applies by treating the interval $[d,d']$ as a block of compression increments. To enable this analysis, we assume the optimal error curve $d \to \|x - g_d^\star\|$ obeys three natural structural properties, which we state formally below. 

\vspace{-3mm}
\paragraph{Assumption 1 (Monotonicity).} For all $d \geq 0$,
\[\|x - g_d^\star\| \;\geq\; \|x - g_{d+1}^\star\|. \tag{A1} \label{A1}\]
That is, enlarging the parameter budget cannot worsen the optimal fit. 

\vspace{-3mm}
\paragraph{Assumption 2 (Discrete convexity).} For all $d \geq 0$,
\[\|x - g_d^\star\| - \|x - g_{d+1}^\star\| \;\geq\; \|x - g_{d+1}^\star\| - \|x - g_{d+2}^\star\|. \tag{A2} \label{A2}\]
This is a discrete diminishing-returns condition: the marginal improvement obtained by increasing model size is larger at small $d$ and becomes progressively smaller as $d$ increases.

\vspace{-3mm}
\paragraph{Assumption 3 (Zero error).} There exists a finite $d^\star$ such that 
\[\|x - g_{d^\star}^\star\| = 0. \tag{A3} \label{A3}\]
In other words, at some sufficiently large model size, the model $g$ can exactly parameterize $x$.

Assumption~\ref{A3} is used only to anchor the telescoping argument: Theorem~\ref{thm:main_relaxed} in Appendix~\ref{appendix:approximateA3} shows that analogous bounds hold with $d^\star$ replaced by $d^\star(\epsilon)$ satisfying $\|x - g^\star_{d^\star(\epsilon)}\| \leq \epsilon$, at the cost of an additive $\epsilon$ in the error bound. Since $d^\star(\epsilon) \leq d^\star$, this trades an $\epsilon$ error floor against a smaller multiplier, and can be strictly tighter than the exact-fit version with $d^\star$.

Together, Assumptions \ref{A1}--\ref{A3} impose a simple yet powerful geometry on the optimal error curve: it decreases monotonically at a diminishing rate and converges to zero at a finite model size $d^\star$, enabling the gap between optimal models at adjacent compression levels to control the compression error relative to an unknown ground truth. We make this precise in the following sections. 

\vspace{-1mm}
\subsection{Main Result} 
We begin with the idealized setting in which the globally optimal models $g_d^\star$ and $g_{d+1}^\star$ are available. \Cref{thm:main} is the core of our analysis in this setting: it shows that the compression error at parameter budget $d$ can be bounded by the difference between two neighboring optimal compressive representations, scaled by the remaining gap to the exact-fit regime. Put differently, if increasing the model size by one parameter only changes the optimal representation a little, then the model at size $d$ is already close to the ground-truth signal.

\begin{theorem}[Global Error Bound]
\label{thm:main}
Under Assumptions \ref{A1}--\ref{A3}, for any parameter budget $d \leq d^\star$, the compression error is bounded by
\begin{align}
    \|x - g_d^\star\| \leq (d^\star - d)\|g_d^\star - g_{d+1}^\star\|.
\end{align}
\end{theorem}
\begin{proof}
Writing $\|x - g_d^\star\|$ as a telescoping sum and applying Assumption \ref{A3} yields
\[\|x - g_d^\star\| = \sum_{i=1}^{d^\star - d} \|x - g_{d+i-1}^\star\| - \|x - g_{d+i}^\star\|.\]
Assumption \ref{A2} ensures each summand is no larger than the first, so
\[\|x - g_d^\star\| \leq (d^\star - d)(\|x - g_d^\star\| - \|x - g_{d+1}^\star\|) \leq (d^\star - d)\|g_d^\star - g_{d+1}^\star\|,\]
where the last step follows from the triangle inequality $\|g_d^\star - g_{d+1}^\star\| \geq \|x - g_d^\star\| - \|x - g_{d+1}^\star\|$.
\end{proof}
\vspace{-2mm}
Theorem~\ref{thm:main} is model-agnostic and efficiently computable:
it reduces the unknown error $\|x - g_d^\star\|$ relative to the ground truth to an efficiently computable quantity involving only two models at adjacent compression levels, where the scale factor $(d^\star - d)$ can be determined analytically from the architecture specification of $g$. Critically, this bound does not require any knowledge of the target signal $x$, only the difference between two model outputs and the known zero-error budget $d^\star$.

\begin{remark}[Squared-error variant]
\label{remark:squared-error}
For architectures where Assumption \ref{A2} holds only for the squared error, i.e. $\|x - g_d^\star\|^2 - \|x - g_{d+1}^\star\|^2 \;\geq\; \|x - g_{d+1}^\star\|^2 - \|x - g_{d+2}^\star\|^2$, an analogous telescoping argument using Assumption \ref{A1} yields
\[\|x - g_d^\star\| \leq 2(d^\star - d)\|g_d^\star - g_{d+1}^\star\|,\]
a bound that differs from Theorem~\ref{thm:main} only by a constant factor of 2. This variant is useful for architectures whose approximation theory is more naturally expressed in terms of squared residuals. 
\end{remark}

In practice, estimating $\|g_d^\star - g_{d+1}^\star\|$ from a single adjacent pair of models can be sensitive to training instabilities. In Appendix~\ref{appendix:batched}, we provide a batched variant of \Cref{thm:main} that averages the measured gap over a contiguous block of $b$ compression increments (parameter counts), yielding smoother bound estimates at the cost of $b$ additional model fits.

\vspace{-1mm}
\subsection{Handling the Optimization Gap}
Theorem~\ref{thm:main} and its batched variant are stated in terms of globally optimal models, meaning the minimization in \Cref{eq:gap_definition} is solved exactly. In practice, compressive parameterizations $g$ are often fit by nonconvex optimization, making it difficult to ensure that the trained model at a given parameter count is globally optimal. To account for this gap between theory and practice, let $\tilde{g}_d$ denote the output of a model trained to practical convergence at size $d$, and define the optimization error $\eta_d \coloneq \tilde{g}_d - g_d^\star$ as the gap between the trained model output and optimal size-$d$ model output. \Cref{thm:main_opt_gap} shows that our global error bound extends cleanly to this setting, with the bound degrading in proportion to the size and consistency of optimization errors across compression levels.

\def\linegap{-0.35mm}
\def\figurewidth{0.993}
\begin{figure}[ht!]
\vskip -3mm
\centering
\includegraphics[width=\figurewidth\linewidth]{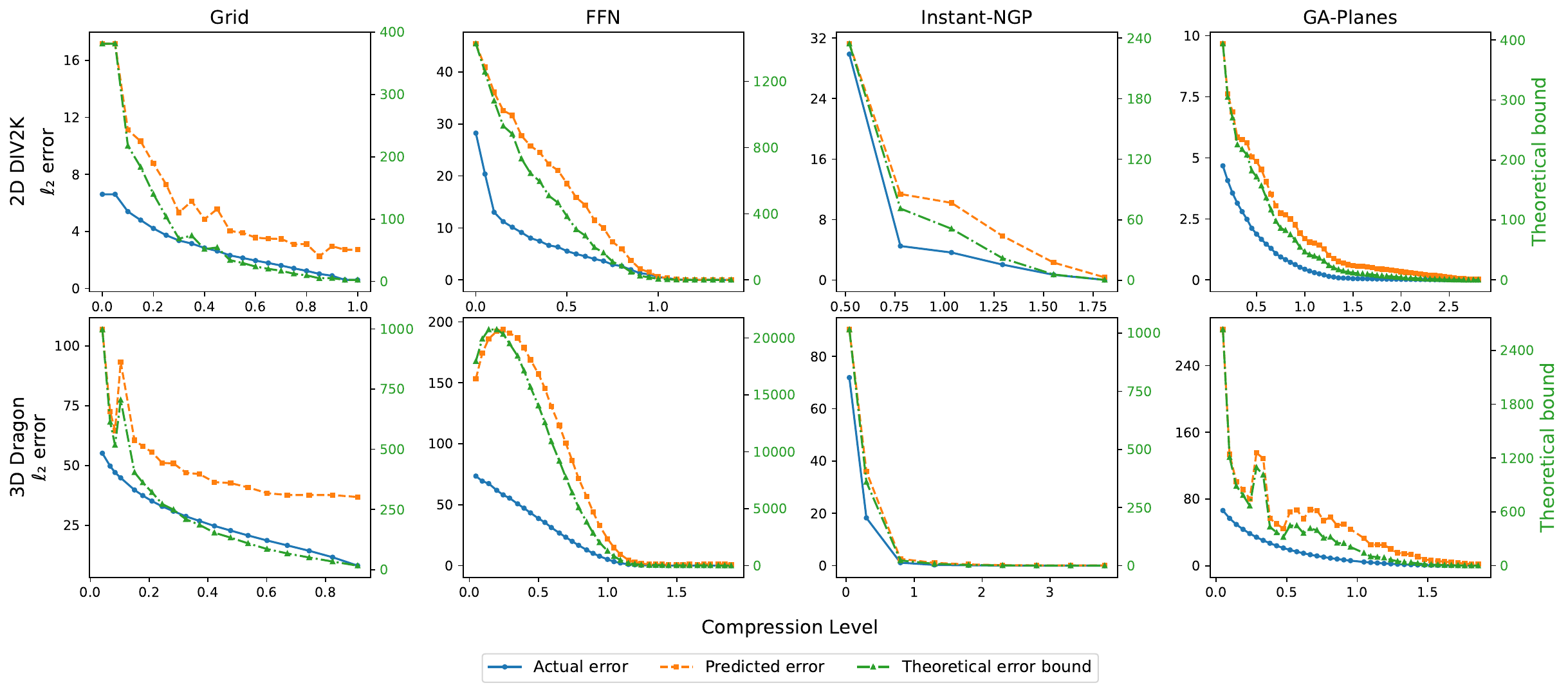}
\vskip 1mm
\begin{minipage}{0.98\textwidth}
    \centering
    \caption{\textbf{Global error curves for direct fitting of signals.} 
    We fit Grid, FFN, Instant-NGP, and GA-Planes directly to a 2D natural image from DIV2K~\citep{agustsson2017div2k} and the 3D Stanford Dragon surface volume~\citep{curless1996dragon}. The actual error (blue, left axis) is the $\ell_2$ distance between reconstruction and ground truth, the optimization-gap-adjusted theory-derived bound (green, right axis) follows from Theorem~\ref{thm:main_opt_gap} as a function of compression level, and the predicted error (orange, left axis) rescales this bound by the empirical constant $c$ applied to the per-increment quantity $B(d)$ (Section~\ref{subsec:constant-c}). Although the theory-derived bounds are conservative by orders of magnitude, the calibrated predictions track the actual error across architectures, signals, and dimensionalities.} 
    \label{fig:task1_error_curves}
\end{minipage}
\vspace{-5mm}
\end{figure}

\begin{theorem}[Error Bound with Optimization Gap]
\label{thm:main_opt_gap}
Under Assumptions \ref{A1}--\ref{A3}, for any converged model outputs $\tilde{g}_d, \tilde{g}_{d+1}$ with optimization errors $\eta_d, \eta_{d+1}$, the compression error is bounded by
\begin{align}
    \|x - \tilde{g}_d\| \leq (d^\star - d)\|\tilde{g}_d - \tilde{g}_{d+1}\| + (d^\star - d)\|\eta_{d+1} - \eta_d\| + \|\eta_d\| .
\end{align}
\end{theorem}
The three terms in Theorem~\ref{thm:main_opt_gap} admit clear interpretations. The first is the same neighboring-model discrepancy that appears in Theorem~\ref{thm:main} and is directly computable. The second penalizes optimization inconsistency across adjacent compression levels: if one model is trained substantially better than its neighbor, the observed gap $\|\tilde{g}_d - \tilde{g}_{d+1}\|$ becomes a poor surrogate for $\|g_d^\star - g_{d+1}^\star\|$. The third is the irreducible optimization error at the parameter budget of interest. When $\eta_d =  \eta_{d+1} = 0$, the bound in \Cref{thm:main_opt_gap} reduces exactly to the bound in \Cref{thm:main}. More broadly, this result highlights that the local gap between nearby trained models remains a useful estimate of the true compression error as long as optimization behavior does not vary too erratically across adjacent compression levels.

\vspace{-1mm}
\subsection{Adaptation to Upper/Lower Bound Envelopes}
For some model families, Assumption \ref{A2} cannot be readily verified for the true optimal error $\|x - g_d^\star\|$. In such cases, it suffices to replace the optimal error curve with tractable upper and lower envelopes that preserve the desired geometry. This allows us to retain the same general strategy to bound the compression error, at the cost of an additive slack term reflecting the width of the envelope. Concretely, suppose there exist functions $L,U:\mathbb{Z}_{+} \to \mathbb{R}_+$ satisfying
\[L(d) \leq \|x - g_d^\star\| \leq U(d) \;\;\forall d, \;\;\;\;\; U(d^\star) = L(d^\star) = 0,\]
where the upper bound $U$ is monotone and discretely convex. The lower envelope $L$ need not be convex; it serves only to quantify how loose the upper envelope may be at each parameter budget. 

\begin{theorem}[Error Bound via Envelopes]
\label{thm:main_envelope}
Suppose the error curve satisfies Assumptions \ref{A1} and \ref{A3}, and suppose there exist envelopes $L$ and $U$ satisfying the conditions above, with $U$ monotone and discretely convex. Then the compression error is bounded by
\begin{align}
    \|x - g_d^\star\| \leq (d^\star - d)\big(\|g_d^\star - g_{d+1}^\star\| + U(d) - L(d)\big).
\end{align}
With optimization gap, 
\begin{align}
    \|x - \tilde{g}_d\| \leq (d^\star - d)\big(\|\tilde{g}_d - \tilde{g}_{d+1}\| + \|\eta_{d+1} - \eta_d\| + U(d) - L(d) \big) + \|\eta_d\|.
\end{align}
\end{theorem}

Theorem~\ref{thm:main_envelope} is important conceptually because it broadens the scope of the framework in \Cref{thm:main} by showing that discrete convexity of the true error curve is not strictly necessary. What matters is that the error curve can be bounded between surrogates with appropriate properties. The slack term $U(d) - L(d)$ quantifies the price of this relaxation: when $U = L$, the bound reduces to Theorem~\ref{thm:main}, while looser envelopes yield valid but more conservative bounds on the compression error.

\vspace{-1mm}
\subsection{Verification for Representative Architectures}
The previous theorems are abstract and apply to any signal parameterization whose error curve satisfies the stated structural assumptions. We now verify that the four representations we study satisfy the appropriate conditions, and identify which theorem governs each representation. The four models are: a Grid baseline (direct interpolated pixel/voxel parameterization), Fourier feature networks (FFN) \cite{tancik2020ffn}, Instant-NGP \cite{muller2022ingp}, and 2D linear-decoder GA-Planes \cite{sivgin2024gaplanes} (see the GA-Planes formulations in Appendix~\ref{appendix:proofs:gap}). 

\begin{theorem}
\label{thm:model-specific}
Let $x \in \mathbb{R}^n$. Consider Grid, FFN and Instant-NGP (INGP) with one-hidden-layer ReLU decoders, and a 2D GA-Planes (GAP) with a linear decoder. Then:
\begin{enumerate}
    \item All four models satisfy Assumptions~\ref{A1} and~\ref{A3}, with the FFN~\ref{A3} guarantee holding with high probability over its random Fourier features.

    \item Grid satisfies the squared-error version of Assumption~\ref{A2} when $x$ has radially nonincreasing power spectrum, while GAP satisfies it unconditionally. Thus,~\Cref{remark:squared-error} applies.

    \item Under the assumptions of~\citet{siegel2020nnupperbound} and~\citet{achour2022nnlowerbound}, FFN admits monotone, discretely convex upper bound $U(d)$ and lower bound $L(d)$ with $L(d) \leq \|x-g_d^\star\| \leq U(d)$. Under uniform, level-independent hashing, INGP admits such envelopes for its expected error. Hence, Theorem~\ref{thm:main_envelope} applies.
\end{enumerate}
\end{theorem}

Proofs and additional theoretical results are provided in Appendices~\ref{appendix:proofs} and~\ref{appendix:extensions}.

\vspace{-1mm}
\section{Experiments} \label{sec:exps}
\def\linegap{-0.40mm}
\def\figurewidth{0.991}
\begin{figure}[ht!]
\vskip -1mm
\centering
\includegraphics[width=\figurewidth\linewidth]{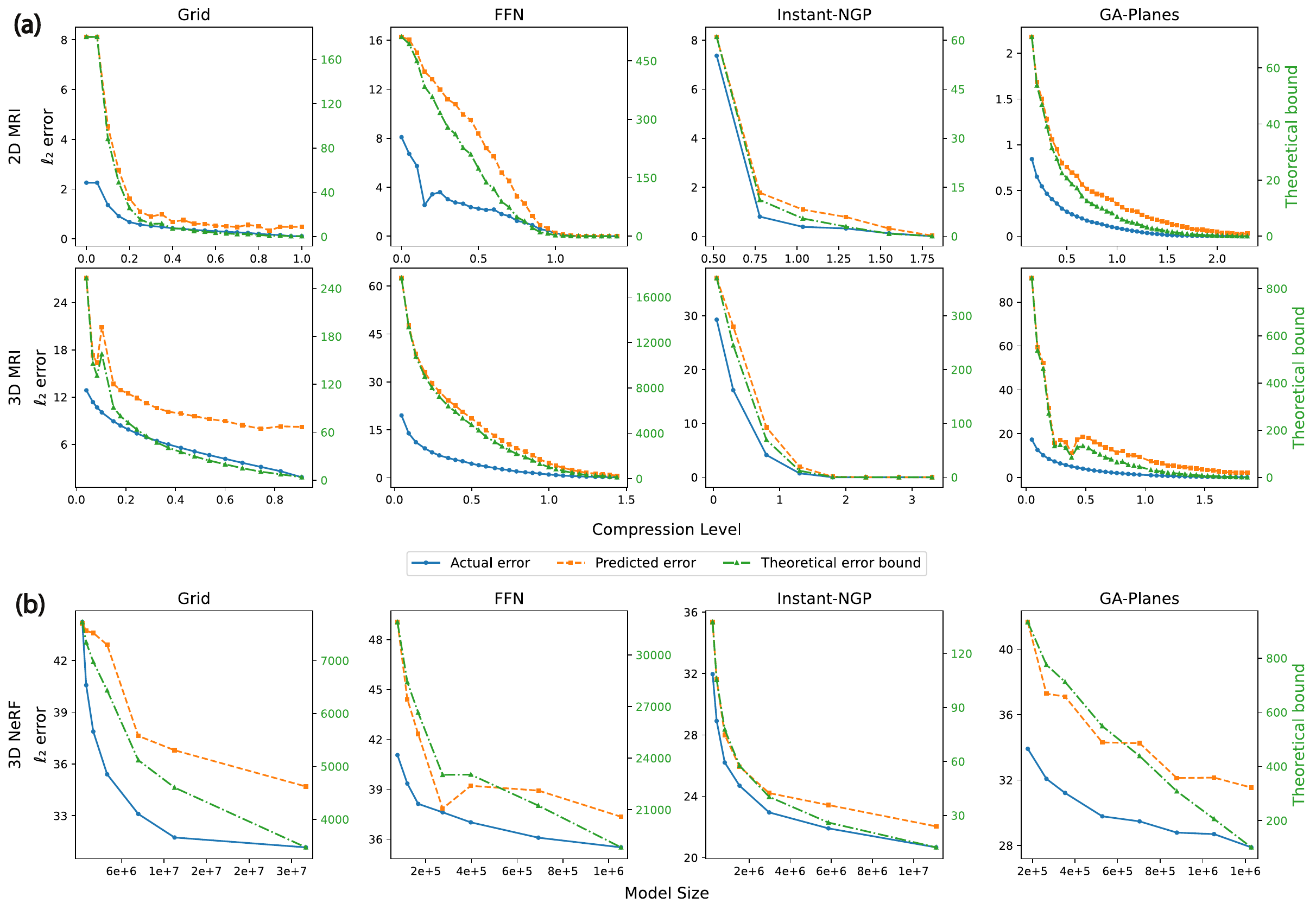}
\vskip 1mm
\begin{minipage}{0.97\textwidth}
    \centering
    \caption{\textbf{Global error curves for inverse problems.} 
    We evaluate Grid, FFN, Instant-NGP, and GA-Planes on three imaging tasks across two settings: (a) on 2D and 3D MRI reconstruction, and (b) radiance field reconstruction.  MRI results are plotted against compression level $d/n$ while NeRF results are plotted against model size, as the ground-truth signal dimension is not well-defined for radiance fields. For MRI, the actual error (blue, left axis) is the $\ell_2$ distance between reconstruction and ground truth; for NeRF, it is the average $\ell_2$ distance between rendered test views and the corresponding ground-truth test views. The optimization-gap-adjusted theory-derived bound (green, right axis) follows from Theorem~\ref{thm:main_opt_gap} as a function of model size, and the predicted error (orange, left axis) rescales this bound by the constant $c$ applied to the per-increment quantity $B(d)$ (Section~\ref{subsec:constant-c}). Despite the indirect supervision of these inverse problems, the predicted error continues to track the actual reconstruction error across compression levels and model sizes.}
    \label{fig:task2_error_curves}
\end{minipage}
\vspace{-5mm}
\end{figure}

\begin{figure}[ht!]
\vskip -3mm
\centering
\includegraphics[width=\figurewidth\linewidth]{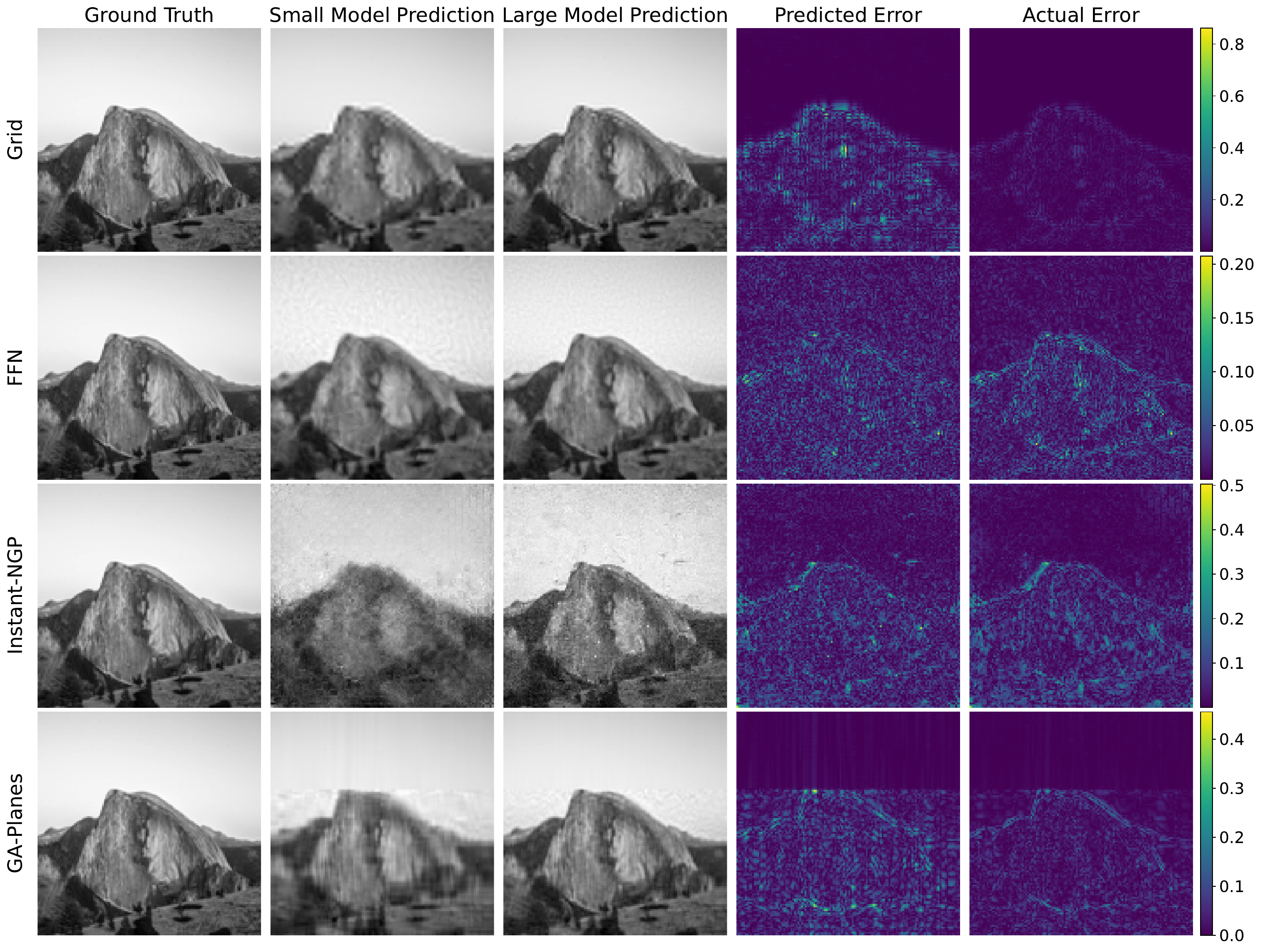}
\vskip 1mm
\begin{minipage}{0.98\textwidth}
    \centering
    \caption{\textbf{Local error heatmaps for direct fitting of a DIV2K~\cite{agustsson2017div2k} natural image.} Each row corresponds to a different architecture, with small and large models at compression levels $d/n \approx 0.2$ and $0.4$, respectively. Across all four representations, the predicted error heatmap $c|\tilde{g}_{\text{small}} - \tilde{g}_{\text{large}}|$ captures the spatial structure of the actual error heatmap $|\tilde{g}_{\text{small}} - x|$, showing that local compression artifacts can be estimated efficiently without ground-truth access.}
    \label{fig:2D_div2k_heatmaps}
\end{minipage}
\vspace{-5mm}
\end{figure}

\begin{figure}[ht!]
\vskip -3mm
\centering
\includegraphics[width=\figurewidth\linewidth]{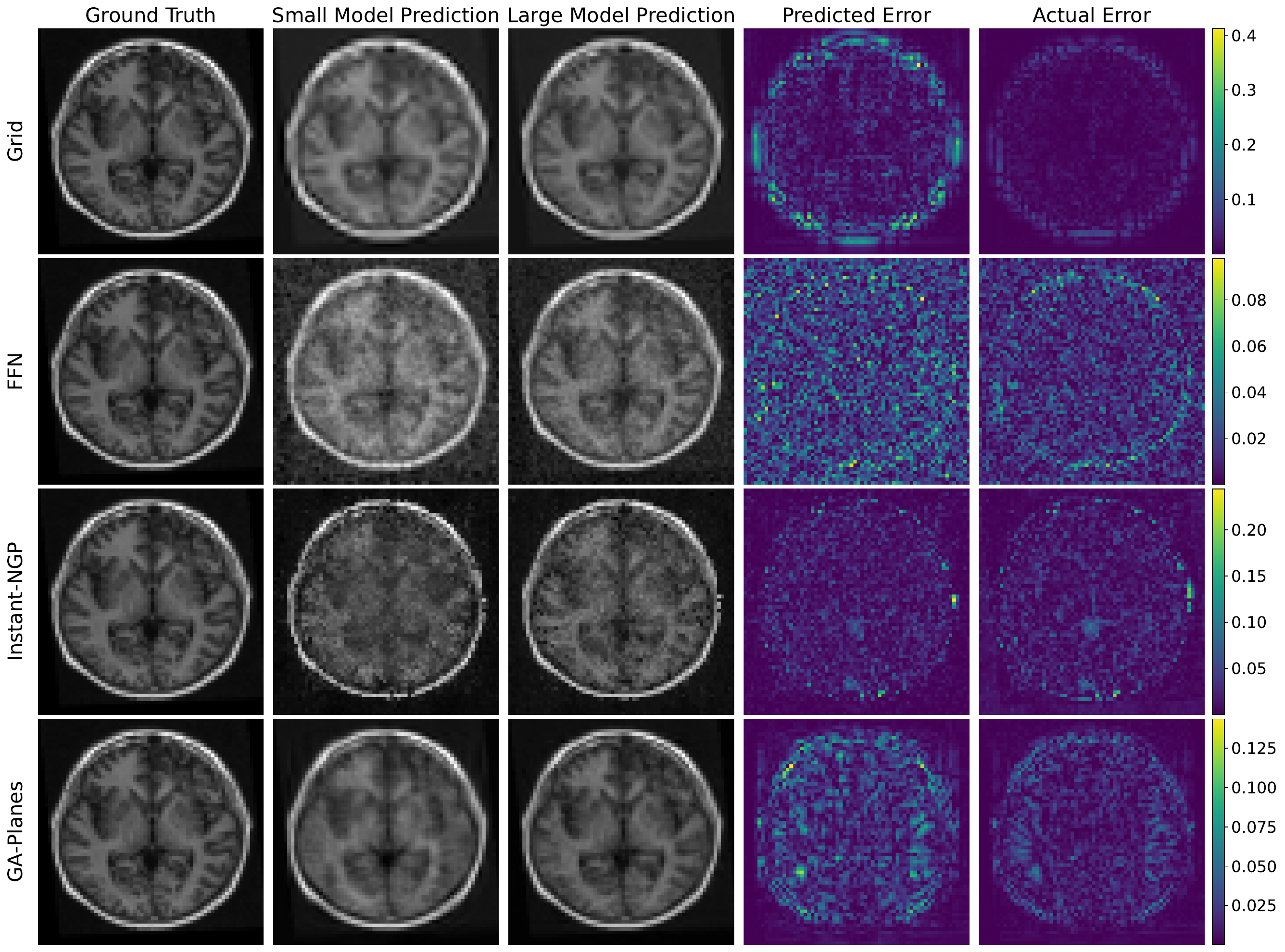}
\vskip 1mm
\begin{minipage}{0.98\textwidth}
    \centering
    \caption{\textbf{Local error heatmaps for 3D MRI volume reconstruction.} 
    Each row corresponds to a different architecture, with small and large models at compression levels $d/n \approx 0.2$ and $0.4$, respectively. Each model maps 3D coordinates to signal values and is supervised indirectly through a $k$-space loss on a 3D volume from the ATLAS dataset~\cite{liew2021atlas}. A single axial slice is shown for visualization. Across all four representations, the predicted error heatmap $c|\tilde{g}_{\text{small}} - \tilde{g}_{\text{large}}|$ captures the spatial structure of the actual error heatmap $|\tilde{g}_{\text{small}} - x|$, indicating that reconstruction artifacts can be localized efficiently without ground-truth access, even for 3D inverse problems.}
    \label{fig:3D_mri_heatmaps}
\end{minipage}
\vspace{-5mm}
\end{figure}

\begin{figure}[ht!]
\vskip -3mm
\centering
\includegraphics[width=\figurewidth\linewidth]{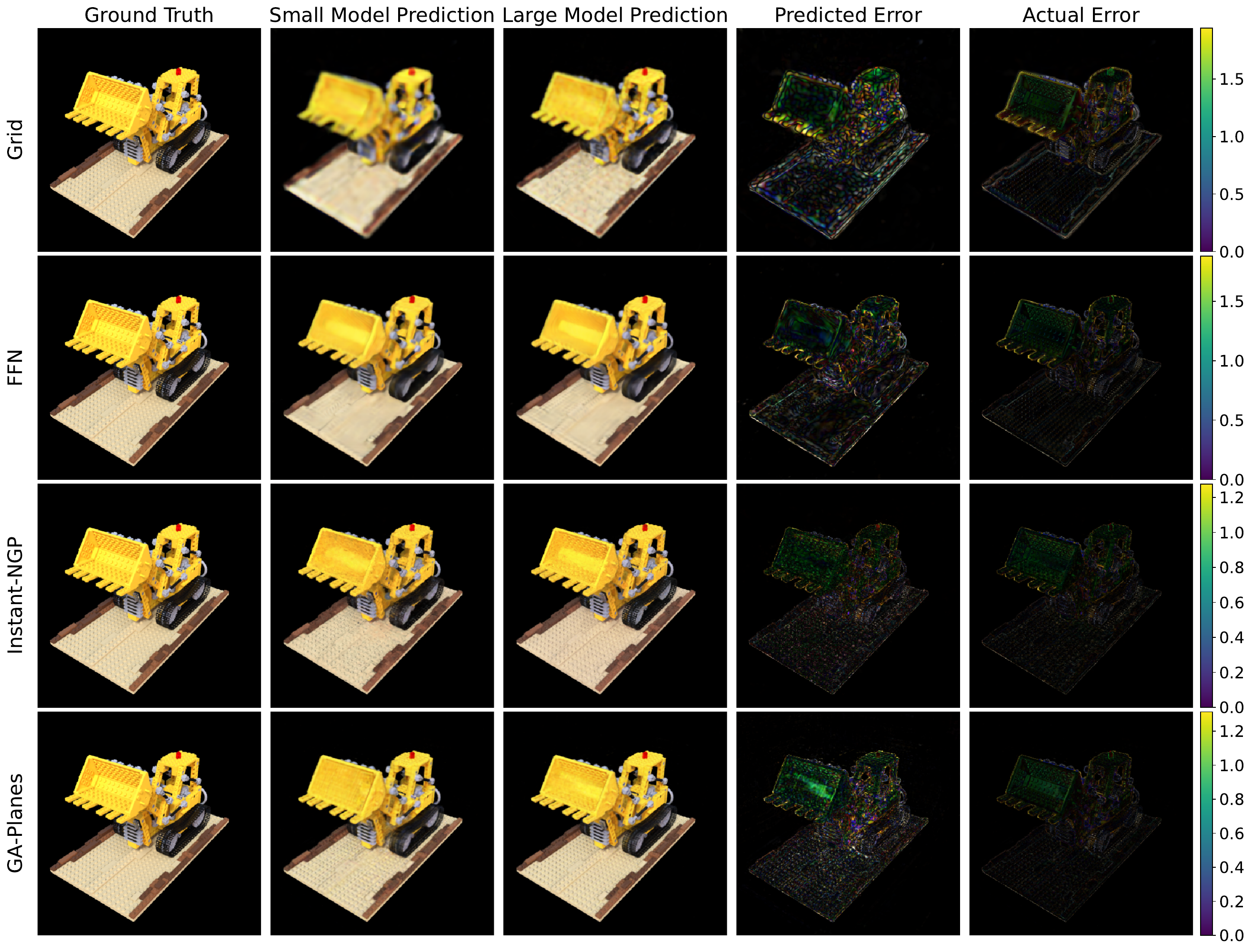}
\vskip 1mm
\begin{minipage}{0.97\textwidth}
    \centering
    \caption{\textbf{Local error heatmaps for 3D radiance field reconstruction.} Each row corresponds to a different architecture, with small and large models having $2.5\times10^5$ and $1.6\times 10^6$ parameters, respectively. Each model represents a 3D radiance field supervised through a photometric loss on training views from the Lego scene in NeRF-Blender~\citep{mildenhall2020nerf}; a single test viewpoint is shown. Across all four representations, the predicted error heatmap $c|\tilde{g}_{\text{small}} - \tilde{g}_{\text{large}}|$ captures the spatial structure of the true error heatmap, indicating that reconstruction artifacts can be localized efficiently at novel viewpoints without access to the ground-truth test view.}
    \label{fig:3D_nerf_heatmaps}
\end{minipage}
\vspace{-5mm}
\end{figure}

\subsection{Calibrating worst-case bounds into practical predictors} \label{subsec:constant-c}

The bounds in Section~\ref{sec:theoreticalres} are provably valid and distribution-free, but conservative in practice. The factor $(d^\star - d)$ is worst-case tight under Assumptions \ref{A1}--\ref{A3} alone (Theorem~\ref{thm:factor_tightness}, Appendix~\ref{appendix:factor_tightness}), so tighter bounds require further assumptions. If Assumption~\ref{A2} is strengthened to a successive geometric decay of marginal error at rate $\rho < 1$, Theorem~\ref{thm:geometric} shows that the multiplier $(d^\star - d)$ can collapse to $(1-\rho)^{-1}$. Motivated by this refinement and the observation that the local difference between compressed reconstructions at adjacent model sizes already captures the shape of the true error curve in practice, we use the theorems as a structural template and introduce a single architecture-specific constant $c$ that converts our theoretical bound into a practical, signal-adaptive predictor of compression error. For $d < d^\star$, define the per-increment quantity
\[B(d) \coloneq \|\tilde{g}_d - \tilde{g}_{d+1}\| + \|\eta_{d+1} - \eta_d\| + \frac{\|\eta_d\|}{d^\star - d},\]
so that Theorem~\ref{thm:main_opt_gap} reads $\|x - \tilde{g}_d\| \leq (d^\star - d)B(d)$. For each architecture, we choose the smallest $c$ such that the predictor $cB(d)$ dominates the observed errors on a calibration signal from the ATLAS dataset~\citep{liew2021atlas}, then hold it fixed across subsequent signals, datasets, and compression levels. 
Note that we calibrate separately for radiance fields using the Lego scene in NeRF-Blender~\citep{mildenhall2020nerf}, with the constants listed in Table~\ref{tab:c_values}.
We quantitatively evaluate the stability of this constant $c$ across 30 signals in Appendix~\ref{appendix:calibration-sensitivity}.

\vspace{-1mm}
\subsection{Estimating the optimization gap}
\label{subsec:opt-gap}
Theorem~\ref{thm:main_opt_gap} introduces the optimization error $\eta_d = \tilde{g}_d - g_d^\star$, the gap between the converged reconstruction and the global optimum at compression level $d/n$. Since $g_d^\star$ is unobservable, we estimate $\|\eta_d\|$ from restart variability. For a fixed signal, architecture, and compression level, we obtain the reconstructions $\tilde{g}_d^{(1)}, \tilde{g}_d^{(2)}, \ldots, \tilde{g}_d^{(m)}$ by training the model $m$ times from independent random initializations with identical hyperparameters and convergence criteria, and compute the maximum pairwise $\ell_2$ distance $h_m = \max_{i\neq j} \|\tilde{g}_d^{(i)} - \tilde{g}_d^{(j)}\|$. Since $h_m$ is nondecreasing in $m$, we add restarts until $|h_m - h_{m-1}|$ falls below a tolerance, and use the resulting plateau value as a practical estimate of a bound on $\|\eta_d\|$ at compression level $d/n$.

The intuition is that variability across independently converged solutions provides an observable proxy for how far practical training may be from the ideal optimum. If at least one restart is close to $g_d^\star$, then $h_m$ controls the optimization error of every other restart via the triangle inequality; if all runs converge to similar outputs, the estimated $\|\eta_d\|$ is small. This estimate cannot be certified as an upper bound on $\|\eta_d\|$ since multiple runs may still converge to the same suboptimal minimum, but it provides a practical, ground-truth-free diagnostic for optimization stability. We apply this correction to the global error curves in Figures~\ref{fig:task1_error_curves} and~\ref{fig:task2_error_curves} (Appendix~\ref{appendix:exp:global}), which therefore report the full optimization-gap-adjusted bound of Theorem~\ref{thm:main_opt_gap}. In Appendix~\ref{appendix:computational-overhead}, we show that the unadjusted adjacent-model predictor requires $1.99\times$--$2.01\times$ the computation time of fitting a single model, while the restart-adjusted predictor used for our global error results requires $5.09\times$--$6.78\times$.

\vspace{-1mm}
\subsection{Global error curves}
Figures~\ref{fig:task1_error_curves} and~\ref{fig:task2_error_curves} evaluate the calibrated adjacent-model predictor across direct fitting and inverse problem settings, varying the primary capacity parameter of each architecture while holding other hyperparameters fixed. Across all settings, the predicted error captures the overall shape and scale of the true error curves across architectures, datasets, and tasks. The theoretical bounds remain valid but conservative, whereas the calibrated predictions are substantially tighter and preserve the qualitative decay of the true reconstruction error. The actual error curves are also approximately convex and monotonically decreasing up to optimization variability, providing empirical support for the structural assumptions underlying our theory. These results suggest that adjacent-model differences provide a stable global proxy for reconstruction error in both direct fitting and inverse-problem settings, despite the latter involving indirect supervision and nontrivial forward models, and that the looseness of the worst-case bound is likely architecture-dependent rather than signal-dependent or task-dependent. In Appendix~\ref{appendix:calibration-sensitivity}, we repeat the full model capacity sweep on 30 signals and find that the mean predicted-to-actual error ratio lies between $1.91$ and $2.54$, confirming the approximate factor-of-two conservatism visible in Figures~\ref{fig:task1_error_curves} and~\ref{fig:task2_error_curves}. Signal-specific error prediction also motivates ground-truth-free model-size selection under a prescribed fidelity target; Appendix~\ref{appendix:exp:selection} provides an illustrative MRI example in which statistical (dataset-based) error estimates can underestimate error on new or challenging signals, while our calibrated error estimates adapt and remain conservative across signals.

We note that the NeRF experiments extend beyond our formal shallow-decoder theory in two ways: the decoders are multi-layer nonlinear MLPs (rather than the single-hidden-layer decoders covered by our proofs), and they take additional view-direction inputs to model view-dependent radiance, as is standard in NeRF~\citep{mildenhall2020nerf}. Full architectural details are in Appendix~\ref{appendix:implementation}. These experiments are therefore intended as empirical stress tests of whether the adjacent-model predictor remains useful for more challenging tasks like radiance field reconstruction. Additional details on test signals and global error curves for synthetic signals are provided in Appendices~\ref{appendix:test_signals} and~\ref{appendix:exp:global}. 

\vspace{-1mm}
\subsection{Local error heatmaps}
Although our theory is stated for global $\ell_2$ error, the same adjacent-model principle can be applied pointwise to obtain spatially resolved error estimates. Given a small and large model from the same representation family, we compute the predicted error heatmap $c|\tilde{g}_{\text{small}}-\tilde{g}_{\text{large}}|$ using the constant $c$ from Section~\ref{subsec:constant-c}. This estimate requires only the two trained reconstructions and no access to ground truth. For evaluation only, we compare it against the actual error heatmap $|\tilde{g}_{\text{small}}-x|$. 

We evaluate local error prediction for  direct fitting on 2D natural images from DIV2K~\cite{agustsson2017div2k} (Figure~\ref{fig:2D_div2k_heatmaps}), 3D MRI volume reconstruction from $k$-space measurements using the ATLAS dataset~\cite{liew2021atlas} (Figure~\ref{fig:3D_mri_heatmaps}), and 3D radiance field reconstruction of the Lego scene~\cite{mildenhall2020nerf} (Figure~\ref{fig:3D_nerf_heatmaps}). Across all architectures and tasks, regions of high predicted error consistently coincide with high actual error, while low-error regions are correctly identified. We emphasize that these heatmaps carry no formal theoretical guarantee, as our bounds are stated for the global $\ell_2$ norm, not pointwise or region-wise error. However, their consistency suggests that adjacent-model differences are informative not only in aggregate, but also about the spatial locations where a compressed reconstruction is likely to be unreliable. Further implementation details, additional results for other tasks, and quantitative metrics for the error heatmaps may be found in Appendices~\ref{appendix:implementation} and~\ref{appendix:exp:local}. 

\vspace{-1mm}
\section{Discussion} \label{sec:disc}
In this work, we present, to our knowledge, the first framework for non-asymptotic, signal-specific error bounds for compressive signal parameterizations that are computable without access to ground truth. Under three structural conditions on the optimal error curve, which we show to hold for representative signal parameterizations, we prove that the reconstruction error at a given model size is bounded by the difference between compressed reconstructions at adjacent model sizes. We complement our theory with a simple scalar adaptation that yields predicted error curves closely tracking the true error across diverse applications. We further find empirically that the same local-comparison principle produces informative pixel-level error heatmaps, despite our theory being stated only for global error.

\vspace{-1mm}
\paragraph{Limitations and future work.} Our theoretical guarantees rely on structural assumptions about the optimal error curve that we verify only for the model families studied here, each using shallow, scalar-output MLP decoders. Extending these results to deeper, vector-output decoders and a broader set of architectures and input encodings is an important direction for future work. Our scalar adaptation from worst-case bounds to practical predictions is estimated via a simple heuristic calibration; a more principled calibration procedure could yield tighter error estimates. Finally, our theoretical results bound global reconstruction error; the local heatmaps remain empirical error estimates rather than proven error bounds. Extending our theory to local seminorms or region-of-interest metrics would provide theoretical grounding for the spatially resolved predictions that our experiments suggest are informative, and would strengthen the connection to safety-critical applications where localized compression artifacts can matter more than average error.

\bibliography{ref_arxiv}
\bibliographystyle{unsrtnat}
\clearpage
\appendix

\section*{Appendices}
\paragraph{Norm notation.} Throughout the appendices, we use $L^p$ to denote the function-space norms on continuous domains, consistent with the approximation theory literature~\citep{barron1993universal, siegel2020nnupperbound, achour2022nnlowerbound}, whereas $\ell_p$ denotes the finite-dimensional norm on sampled signals used in the main text. On the finite sample grid used in our experiments, the $L^2$ and $\ell_2$ norms coincide up to a fixed normalization, so the bounds we adopt translate directly to the $\ell_2$ norms used in the main text.
\section{Related Work}
We provide here an expanded discussion of the four lines of prior work summarized in Section~\ref{sec:relworks}.

\subsection{Representative approaches to compressive parameterizations}
\label{appendix:relworks:1}
Compressive, continuous signal representations have found broad application in computational imaging tasks where memory and compute are tightly constrained, including accelerated MRI~\citep{ong2020extrememri, heckel2024deeplearningmri,lustig2007compressedsensing, molaei2023inrmedicalimaging}, robotic perception~\citep{wang2025nerfroboticssurvey, araujo2024robotoicperceptiohn, zhu2022niceslam, sucar2021imap}, remote sensing~\citep{gunasheela2018satellite, gomes2025geospatialcompression, louys1999astronomical, offringa2016compression}, and edge-constrained imaging~\citep{ngo2025dnnedge, yin2025surveyedge, delprete2025optimizing}. 
Recent architectures for these tasks have converged around three design paradigms: coordinate-based neural networks~\citep{chen2021learning, lindell2021bacon, saragadam2023wire, sitzmann2020siren, tancik2020ffn, strumpler2022implicit}, trainable spatial encodings~\citep{muller2022ingp}, and tensor-factorized models~\citep{chen2022tensorf,fridovichkeil2023kplanes,sivgin2024gaplanes}. Fourier Feature Networks (FFNs)~\citep{tancik2020ffn} preprocess input coordinates with a sinusoidal feature lifting before passing them to a ReLU MLP, transforming the effective neural tangent kernel into a stationary kernel with tunable bandwidth and improving the learning of high-frequency structures. Instant-NGP~\citep{muller2022ingp} stores learnable features in a multi-resolution hash table indexed by input coordinates, and decodes the resulting concatenated multi-scale features with a lightweight MLP. This multi-resolution construction helps the decoder disambiguate hash collisions, enabling high fidelity within a compact parameter budget. Geometric Algebra Planes (GA-Planes)~\citep{sivgin2024gaplanes} combines line, plane, and volume features with a neural decoder, and admits convex and semi-convex optimization formulations that in the 2D setting are provably equivalent to a (masked) low-rank plus low-resolution matrix factorization. Together, these models span three qualitatively different compression mechanisms---spectral lifting, multi-resolution feature sharing, and structured low-dimensional factorization---and approximately sample the current Pareto frontier of model size, expressiveness, and optimizability~\citep{sivgin2024gaplanes}. We therefore focus our theoretical and experimental analysis on these three models as well as an interpolated Grid baseline~\citep{kim2025gridsoutperforminrs}.

\subsection{Rate-distortion theory and compression bounds}
\label{appendix:relworks:2}
Classical rate--distortion theory characterizes the fundamental tradeoff between compression rate and reconstruction fidelity for stochastic sources, providing a lower bound on the distortion achievable at any given bit budget~\citep{shannon1959probability, cover2006elements}. A growing line of work has brought rate--distortion ideas into implicit neural representation-based compression pipelines for images, volumes, and video, typically by combining trained neural fields with quantization, entropy coding, or related coding mechanisms~\citep{strumpler2022implicit, chen2021nerv, dai2025implicit, guo2023bayesianinrs}. A complementary direction studies the model size required for lossless neural representation:~\citet{han2025losslessinr} derive explicit upper bounds on the number of parameters needed for an INR to achieve a given bit precision on a discrete signal, showing that this threshold grows exponentially with bit depth. These works make important progress toward INR-based compression, but their primary emphasis remains codec or representation design, and evaluation is largely empirical rather than oriented around signal-specific theoretical guarantees on representation error. 

Our setting differs from this literature in how representation size is measured. Rate--distortion theory and neural codecs measure size in bits, while we consider real-valued parameter vectors $\theta \in \mathbb{R}^d$ and measure size by parameter count $d$. These two notions of size are not interchangeable: bit-level guarantees do not characterize the error achievable by a fixed architecture at a given $d$ and accordingly, parameter count alone does not specify a bit rate without additional choices of numerical precision and coding. Moreover, classical rate--distortion theory assumes a known or estimable probability distribution over sources~\citep{cover2006elements}, whereas our goal is to bound the representation error of a single signal at finite model size without distributional assumptions. Our results therefore complement bit-level analyses: a fitted parameterization may subsequently be quantized and entropy-coded, a step our framework does not model. To our knowledge, no prior work provides non-asymptotic, signal-specific error bounds for compressive representations as a function of model size, that are both theoretically justified and computable without access to ground truth. Our framework addresses this gap by replacing distributional assumptions with three structural properties of the optimal error curve, which we verify for each architecture we study.

\subsection{Approximation rates for neural networks}
\label{appendix:relworks:3}
A substantial body of theoretical work characterizes how accurately neural networks of fixed size can approximate functions drawn from classical function spaces. Early results of~\citet{barron1993universal} established dimension-independent $L^2$ rates for two-layer networks approximating functions with bounded spectral norm.~\citet{siegel2020nnupperbound} generalized these results to broader activation families, providing approximation rates for polynomially decaying non-sigmoidal activations and any bounded integrable activations, and characterized how the error scales with network width and depth.~\citet{achour2022nnlowerbound} complemented these upper bounds with tight lower bounds for piecewise-polynomial feed-forward networks on H\"{o}lder balls and multivariate monotone functions. A structural implication shared by both works is that the approximation error of an MLP is bounded above by a convex, decreasing function of the number of parameters, a key property that we leverage in our main theorem. In particular, we adopt the upper bound of~\citet{siegel2020nnupperbound} and the lower bound of~\citet{achour2022nnlowerbound}, adapted to account for the Fourier feature embedding of FFNs, as the convex envelope for the error curve in our proof of Theorem \ref{thm:model-specific}.

\subsection{Universal approximation for neural networks}
\label{appendix:relworks:4}
The classical Universal Approximation Theorem~\citep{hornik1989multilayer} guarantees that a sufficiently wide single-hidden-layer MLP can approximate any continuous function on a compact domain, but provides no constructive bound on the width required.~\citet{park2021minimumwidth} resolved this question by identifying $\max\{d_x + 1, d_y\}$ as the exact minimum-width threshold for universal approximation of ReLU networks on $\mathbb{R}^{d_x}$, ruling out any strictly smaller width and establishing this threshold as a fundamental architectural limit. Here $d_x$ denotes the input dimension and $d_y$ the output dimension (number of scalar signal channels predicted by the decoder).~\citet{kim2024tighterminimumwidth} further refined this guarantee for the practically relevant compact-domain setting, proving that the minimum width for $L^p$ approximation of $L^p$ functions on $[0,1]^{d_x}$ drops to $\max\{d_x, d_y, 2\}$ for ReLU-like activations, including ReLU, GELU, and Softplus. These results are the key ingredient in our proof that all three compressive representations with MLP decoders we consider satisfy Assumption \ref{A3}: the existence of a finite model size at which the optimal representation error vanishes. They are invoked most directly in the proof for FFNs, where the coordinate domain is normalized to $[0,1]^k$ and the effective decoder is a ReLU MLP composed with a Fourier feature embedding.

\section{Proofs of Theorems}
\label{appendix:proofs}
\paragraph{Notation for proofs.} 
Proofs follow the same notational conventions as the main text.
For a fixed parameter budget $d$, let $g_d^\star$ denote the globally optimal representation, $\tilde{g}_d$ denote the output of a model trained to practical convergence, and $\eta_d = \tilde{g}_d - g_d^\star$ be the optimization error. 

\subsection{Proof of \Cref{remark:squared-error}}
\begin{proof}
    Suppose Assumption~\ref{A2} holds only for the squared  error, i.e., $\|x-g_d^\star\|^2 - \|x - g_{d+1}^\star\|^2 \geq \|x-g_{d+1}^\star\|^2 - \|x - g_{d+2}^\star\|^2$ for all $d \geq 0$. Writing $\|x - g_d^\star\|^2$ as a telescoping sum and applying Assumption~\ref{A3} yields
    \[\|x - g_d^\star\|^2 = \sum_{i=1}^{d^\star - d} \|x - g_{d+i-1}^\star\|^2 - \|x - g_{d+i}^\star\|^2.\]
    It follows from the assumption that each summand is no larger than the first, so
    \[\|x - g_d^\star\|^2 \leq (d^\star - d)(\|x - g_{d}^\star\|^2 - \|x - g_{d+1}^\star\|^2).\]
    Factoring the difference of squares and applying the triangle inequality $\|x - g_d^\star\| - \|x - g_{d+1}^\star\| \leq \|g_d^\star - g_{d+1}^\star\|$, we obtain 
    \begin{align*}
        \|x - g_{d}^\star\|^2 - \|x - g_{d+1}^\star\|^2 &= (\|x - g_{d}^\star\| - \|x - g_{d+1}^\star\|)(\|x - g_{d}^\star\| + \|x - g_{d+1}^\star\|) \\
        & \leq \|g_d^\star - g_{d+1}^\star\|(\|x - g_{d}^\star\| + \|x - g_{d+1}^\star\|) \\
        & \leq \|g_d^\star - g_{d+1}^\star\|(\|x - g_{d}^\star\| + \|x - g_{d}^\star\|) = 2\|g_d^\star - g_{d+1}^\star\| \|x - g_{d}^\star\|,
    \end{align*}
    where the third line uses Assumption~\ref{A1}. Combining the two inequalities gives
    \[\|x - g_d^\star\|^2 \leq 2(d^\star - d)\|g_d^\star - g_{d+1}^\star\| \|x - g_{d}^\star\|.\]
    If $\|x - g_{d}^\star\|$ = 0, the bound is trivially true; otherwise, dividing both sides by $\|x - g_{d}^\star\|$ yields
    \[\|x - g_d^\star\| \leq 2(d^\star - d)\|g_d^\star - g_{d+1}^\star\|.\]
\end{proof}

\subsection{Proof of Theorem~\ref{thm:main_opt_gap}}
\begin{proof}
By the triangle inequality and the identity $\tilde{g}_d = g_d^\star + \eta_d$, we have
\[\|x - \tilde{g}_d\| = \|x - g_d^\star - \eta_d\| \leq \|x - g_d^\star\| + \|\eta_d\|.\]
    
From Theorem~\ref{thm:main}, we have
\[
\|x - g_d^\star\| \le (d^\star - d)\|g_d^\star - g_{d+1}^\star\|.
\]
Therefore,
\[
\|x - \tilde{g}_d\| \le (d^\star - d)\|g_d^\star - g_{d+1}^\star\| + \|\eta_d\|.
\]
Substituting $g_d^\star = \tilde{g}_d - \eta_d$ and $g_{d+1}^\star = \tilde{g}_{d+1} - \eta_{d+1}$, we obtain
\begin{align*}
    \|x - \tilde{g}_d\|
    &\le (d^\star - d)\|\tilde{g}_d - \eta_d - (\tilde{g}_{d+1} - \eta_{d+1})\| + \|\eta_d\| \\
    &\le (d^\star - d)\Big( \|\tilde{g}_d - \tilde{g}_{d+1}\| + \|\eta_{d+1} - \eta_d\| \Big) + \|\eta_d\| \\
    &= (d^\star - d)\|\tilde{g}_d - \tilde{g}_{d+1}\|
      + (d^\star - d)\|\eta_{d+1} - \eta_d\|
      + \|\eta_d\|,
\end{align*}
which is exactly the claimed bound.
\end{proof}

\subsection{Proof of Theorem~\ref{thm:main_envelope}}
\begin{proof}
Define the slack
\[
\xi_d := U(d) - \|x - g_d^\star\|.
\]
Since $L(d) \le \|x - g_d^\star\| \le U(d)$ by assumption, it follows that
\[
0 \le \xi_d \le U(d) - L(d).
\]

Since $U$ is monotone, discretely convex, and satisfies $U(d^\star) = 0$, the same telescoping and discrete-convexity argument used in the proof of Theorem~\ref{thm:main} applies directly to $U$. In particular,
\[
U(d)
= \sum_{i=1}^{d^\star-d} \bigl(U(d+i-1) - U(d+i)\bigr)
\le (d^\star - d)\bigl(U(d) - U(d+1)\bigr).
\]
It therefore remains to bound the one-step decrement $U(d) - U(d+1)$. Since $\|x - g_{d+1}^\star\| \le U(d+1)$, we have
\begin{align*}
U(d) - U(d+1)
&\le U(d) - \|x - g_{d+1}^\star\| \\
&= \bigl(U(d) - \|x - g_d^\star\|\bigr)
   + \bigl(\|x - g_d^\star\| - \|x - g_{d+1}^\star\|\bigr) \\
&= \xi_d + \|x - g_d^\star\| - \|x - g_{d+1}^\star\| \\
&\le \xi_d + \|g_d^\star - g_{d+1}^\star\|,
\end{align*}
where the last step follows from the triangle inequality. Substituting this bound into the previous inequality and using $\xi_d \le U(d) - L(d)$ yields
\[
U(d)
\le (d^\star - d)\bigl(\|g_d^\star - g_{d+1}^\star\| + U(d) - L(d)\bigr).
\]
Since $\|x - g_d^\star\| \le U(d)$, we can conclude that
\[
\|x - g_d^\star\|
\le (d^\star - d)\bigl(\|g_d^\star - g_{d+1}^\star\| + U(d) - L(d)\bigr).
\]

We now incorporate the optimization gap. By the triangle inequality and the identity $\tilde{g}_d = g_d^\star + \eta_d$,
\[
\|x - \tilde{g}_d\| \le \|x - g_d^\star\| + \|\eta_d\|.
\]
Applying the bound just established for $\|x - g_d^\star\|$ and substituting
\[
g_d^\star = \tilde{g}_d - \eta_d,
\qquad
g_{d+1}^\star = \tilde{g}_{d+1} - \eta_{d+1},
\]
we obtain
\begin{align*}
\|g_d^\star - g_{d+1}^\star\|
&= \|(\tilde{g}_d - \tilde{g}_{d+1}) - (\eta_d - \eta_{d+1})\| \\
&\le \|\tilde{g}_d - \tilde{g}_{d+1}\| + \|\eta_{d+1} - \eta_d\|,
\end{align*}
exactly as in the proof of Theorem~\ref{thm:main_opt_gap}. Combining the preceding inequality gives
\[
\|x - \tilde{g}_d\|
\le (d^\star - d)\bigl(\|\tilde{g}_d - \tilde{g}_{d+1}\| + \|\eta_{d+1} - \eta_d\| + U(d) - L(d)\bigr) + \|\eta_d\|,
\]
which proves the result. 
\end{proof}
\subsection{Proof of Theorem~\ref{thm:model-specific}: Grid}
We prove the Grid representation case by indexing the model family by its spatial resolution rather than by raw parameter count. Let the target signal $x$ be sampled on an $N^D$ lattice (which can be done losslessly for any bandlimited continuous signal, for large enough $N$), and let a Grid model of per-axis resolution $r$ use $d_r = r^D$ parameters, one value at each coarse lattice point. We assume ideal sinc interpolation from the coarse $r^D$ lattice to the fine $N^D$ lattice. Under this interpolation, the set of signals representable by a resolution-$r$ grid is exactly the subspace $\mathcal{S}_r$ of discrete signals whose Fourier coefficients are supported on the hypercube $\{\mathbf{k}: \|\mathbf{k}\|_\infty \leq \lfloor r/2\rfloor\}$. Consequently, the optimal Grid reconstruction $g_{d_r}^\star$ is the orthogonal projection of $x$ onto $\mathcal{S}_r$, and the optimal squared error decomposes cleanly via Parseval's theorem:
\[\|x - g_{d_r}^\star\|^2 = \sum_{\|\mathbf{k}\|_\infty > \lfloor r/2 \rfloor} |X[\mathbf{k}]|^2, \]
where $X[\mathbf{k}]$ denotes the D-dimensional discrete Fourier transform of $x$. This spectral characterization of the optimal error is the foundation for all three lemmas below. 

\begin{gridlemma}[Grid satisfies Assumption~\ref{A1}]\label{lem:grid-a1}
For each per-axis resolution $r \geq 0,$
\[\|x - g_{d_r}^\star\| \geq \|x - g_{d_{r+1}}^\star\|. \]
\end{gridlemma}
\begin{proof}
Since $\lfloor (r+1)/2 \rfloor \geq \lfloor r/2 \rfloor$, the bandlimited subspace $\mathcal{S}_r$ is contained in $\mathcal{S}_{r+1}$. The set of frequencies excluded at resolution $r+1$ is therefore a subset of those excluded at resolution $r$, so
\[\|x - g_{d_r}^\star\|^2 - \|x - g_{d_{r+1}}^\star\|^2 = \sum_{\lfloor r/2 \rfloor < \|\mathbf{k}\|_\infty \leq \lfloor (r+1)/2 \rfloor} |X[\mathbf{k}]|^2 \geq 0 .\]
Since the square root function is monotone increasing, the inequality $\|x - g_{d_r}^\star\|^2 \geq \|x - g_{d_{r+1}}^\star\|^2$ implies $\|x - g_{d_r}^\star\| \geq \|x - g_{d_{r+1}}^\star\|$.
\end{proof}

\begin{gridlemma}[Grid satisfies Assumption~\ref{A2} for squared error, under radial spectral decay]\label{lem:grid-a2}
Assume in addition that the target signal $x$ has a radially nonincreasing power spectrum, meaning $|X[\mathbf{k}]|^2 \geq |X[\mathbf{k'}]|^2$ whenever $\|\mathbf{k'}\|_\infty \geq \|\mathbf{kå}\|_\infty$. Then for each per-axis resolution $r \geq 0$, the optimal square error satisfies the discrete diminishing-returns condition:
\[\|x - g_{d_r}^\star\|^2 - \|x - g_{d_{r+1}}^\star\|^2 \geq \|x - g_{d_{r+1}}^\star\|^2 - \|x - g_{d_{r+2}}^\star\|^2 .\]
\end{gridlemma}
\begin{proof}
From the projection characterization above, the decrease in optimal squared error from resolution $r$ to resolution $r+1$ is exactly the Fourier energy contained in the shell $\lfloor r/2 \rfloor < \|\mathbf{k}\|_\infty \leq \lfloor (r+1)/2 \rfloor$, or equivalently,
\[\|x - g_{d_r}^\star\|^2 - \|x - g_{d_{r+1}}^\star\|^2 = \sum_{\lfloor r/2 \rfloor < \|\mathbf{k}\|_\infty \leq \lfloor (r+1)/2 \rfloor} |X[\mathbf{k}]|^2.\]
Likewise, the decrease from resolution $r+1$ to resolution $r+2$ is exactly the energy in the next shell $\lfloor (r+1)/2 \rfloor < \|\mathbf{k}\|_\infty \leq \lfloor (r+2)/2 \rfloor$, or equivalently,
\[\|x - g_{d_{r+1}}^\star\|^2 - \|x - g_{d_{r+2}}^\star\|^2 = \sum_{\lfloor (r+1)/2 \rfloor < \|\mathbf{k}\|_\infty \leq \lfloor (r+2)/2 \rfloor} |X[\mathbf{k}]|^2.\]
Under the radial spectral decay assumption on $|X[\mathbf{k}]|^2$, the frequencies in the earlier shell carry at least as much power as frequencies in the later shell, so the first shell energy dominates the second. Hence, the one-step improvement in squared error is nonincreasing with resolution, which is exactly discrete convexity for the squared residual.
\end{proof}

\begin{remark} 
The radial spectral decay condition is a natural regularity assumption that is broadly satisfied by many targets of interest in computational imaging, whose power spectral densities typically follow a power-law or exponential decay with frequency~\cite{vanderschaaf1996modelling, torralba2003statistics, lustig2007compressedsensing, xu2019systematic}. Under this common low-pass structure, the energy captured by successively added frequency shells decreases with resolution, so the squared-error discrete convexity established in Lemma~\ref{lem:grid-a2} matches the expected approximation behavior of Grid representations with sinc interpolation. 
\end{remark}

\begin{gridlemma}[Grid satisfies Assumption~\ref{A3}]\label{lem:grid-a3}
There exists a finite per-axis resolution $r^\star$ such that $\|x - g_{d_{r^\star}}^\star\| = 0$.
\end{gridlemma}
\begin{proof}
At full resolution $r^\star = N$, each grid parameter can be set to the corresponding target signal sample. Since ideal sinc interpolation reproduces the data at its own interpolation nodes, the resulting reconstruction agrees with $x$ at every point of the $N^D$ lattice, yielding a zero optimal reconstruction error. The critical parameter budget is therefore $d^\star = N^D$ total parameters, corresponding to per-axis resolution $r^\star = N$.    
\end{proof}

\subsection{Proof of Theorem~\ref{thm:model-specific}: FFN}
Throughout this subsection, we consider an 
FFN \cite{tancik2020ffn} with Fourier feature embedding dimension $q$, a single-hidden layer ReLU MLP of width $w$, and input coordinates $v \in [0,1]^k$. The Fourier feature embedding is
\[\gamma_q(v)\;=\;\big[\;a_1\cos(2\pi b_1^\top v),\;a_1\sin(2\pi b_1^\top v),\;\dots,\;
a_q\cos(2\pi b_q^\top v),\;a_q\sin(2\pi b_q^\top v)\;\big]^\top \in \mathbb{R}^{2q}\]
where $b_l \sim \mathcal{N}(0, \sigma_b^2I_k)$ and $a_l\sim \mathcal{N}(0, \sigma_a^2)$ for $l=1,\ldots,q$. The FFN output is $g_d(\theta) =\mathrm{MLP}_\theta\big(\gamma_q(v)\big)$, where $d = P(L,w,q) := (L-1)w^2 + (2q + L + 1)w + 1$ denotes the total parameter count. In our analysis, we index the model family by the MLP width $w$ while holding the embedding dimension $q$ and number of hidden layers $L=1$ fixed, unless stated otherwise. This is the regime in which the approximation-rate results from Siegel and Xu~\cite{siegel2020nnupperbound} and the lower-bound result of Achour et al.~\cite{achour2022nnlowerbound} can be applied most cleanly. For a target signal $x \in \mathbb{R}^n$ sampled at coordinates $v_1,\ldots,v_n$, the optimal reconstruction error at parameter budget $d$ is
\[\|x - g_d^\star\| = \min_{g_d\in\mathcal{G}_d}\ \|x-g_d\| = \min_{\theta \in \mathbb{R}^d} \ \|x-g_d(\theta)\|,\]
where $\mathcal{G}_d=\{\,g_d(\theta): \theta \in \mathbb{R}^d\}$ is the set of all outputs realizable by an FFN with parameter budget $d$.

\begin{ffnlemma}[FFN satisfies Assumption~\ref{A1}]\label{lem:ffn-a1}
For all $d \geq 0$, 
\[\|x - g_d^\star\| \geq \|x - g_{d+1}^\star\|.\]
\end{ffnlemma}
\begin{proof}
    It suffices to show that $\mathcal{G}_d \subset \mathcal{G}_{d'}$ whenever $d' > d$. We consider three cases corresponding to the three ways of increasing the parameter count. 

    \textit{Case 1: Increasing the embedding dimension: $d'=P(L,w,q+1)$.} An FFN with $q+1$ embedding frequencies can replicate any FFN with $q$ embedding frequencies by appending the two new Fourier features and setting the corresponding incoming weights in the first affine layer to zero. The realized function is unchanged by this operation, so $\mathcal{G}_d \subset \mathcal{G}_{d'}$. 

    \textit{Case 2: Increasing the hidden width: $d'=P(L,w+1,q)$.} An FFN with hidden width $w+1$ can replicate any width-$w$ FFN by padding the first-layer weight matrix with an all-zero row and extending subsequent weight matrices and bias vectors with zeros to maintain consistent dimensions. The additional neuron contributes zero to the output, and the realized function is identical. Thus, $\mathcal{G}_d \subset \mathcal{G}_{d'}$. 

    \textit{Case 3: Increasing the number of hidden layers: $d'=P(L+1,w,q)$.} An FFN with $L+1$ hidden layers can replicate any $L$-layer FFN 
    by inserting an identity layer between any existing ReLU and the next affine layer. The input to the inserted layer is the output of a preceding ReLU and is therefore nonnegative; setting the inserted layer's weights to the identity and its biases to zero, the subsequent ReLU acts as the identity on this nonnegative input. The realized function remains the same, so $\mathcal{G}_d \subset \mathcal{G}_{d'}$. 

    In all three cases, $\mathcal{G}_d \subset \mathcal{G}_{d'}$, so the infimum of the reconstruction error over the strictly larger set $\mathcal{G}_{d'}$ is no greater than over $\mathcal{G}_d$. Therefore, the optimal error is nonincreasing with parameter budget. 
\end{proof}

\begin{ffnlemma}[FFN admits a convex and monotone upper envelope for Assumption~\eqref{A2}]\label{lem:ffn-a2u}
Fix the Fourier feature embedding dimension $q$ and consider one-hidden-layer ReLU FFNs of width $w$. Let $\Omega \subset \mathbb{R}^{2q}$ be a bounded domain containing the range of the embedding $\gamma_q$, and assume there exists a function $f:\Omega \to \mathbb{R}$ in the Barron space $\mathcal{B}^1$ (i.e. $\int_{\mathbb{R}^{2q}}(1+|\omega|)|\hat{f}(\omega)|d\omega$ is bounded) such that $x_i = f(\gamma_q(v_i))$ for all $i = 1,\ldots,n$. Then the approximation error is bounded above by a monotone, discretely convex envelope of the form
\[U(w) = Kw^{-\frac{1}{2}},\]
or equivalently,
\[U(d) = C_ud^{-\frac{1}{2}}.\]
Here $K > 0$ is a constant depending on the function $f$ and the embedding range $\Omega$, but independent of $w$. The constant $C_u > 0$ is the corresponding constant after rewriting the width-dependent bound in terms of the parameter count $d$. 
\end{ffnlemma}
\begin{proof}
    Let $\mathcal{F}_w^{2q}(\sigma)$ be the class of width-$w$ single-hidden-layer ReLU networks on $\mathbb{R}^{2q}$, where $\sigma$ denotes the ReLU activation function. Following the norm convention stated at the beginning of the appendix, $\|f - f_w\|_{L^2(\Omega)}$ agrees with $\|x-g_w(\theta)\|_{\ell_2}$ up to a fixed sample-grid normalization. It follows that the optimal error at width $w$ satisfies
    \[\min_{\theta} \ \|x-g_w(\theta)\| \propto \inf_{f_w \in \mathcal{F}_{w}^{2q}(\sigma)}\|f - f_w\|_{L^2(\Omega)}.\]
    Applying Corollary~\ref{cor:siegelandxu} (see below) with $n_0 = 3$ for the ReLU activation function~\cite{siegel2020nnupperbound}, the $L^2$ approximation error satisfies 
    \[\inf_{f_w \in \mathcal{F}_{w}^{2q}(\sigma)}\|f - f_w\|_{L^2(\Omega)} = \inf_{f_w \in \mathcal{F}_{w}^{2q}(\sigma)}\|f - f_w\|_{H^0(\Omega)} \leq |\Omega|^{\frac{1}{2}}C(c, 0, \text{diam}(\Omega), \sigma)\sqrt{3}\|f\|_{\mathcal{B}^1}w^{-\frac{1}{2}},\]
    for some $c>0$. We can rewrite this as
    \[\min_{\theta} \ \|x-g_w(\theta)\| \leq U(w) = K w^{-\frac{1}{2}},\]
    for some $K > 0$ that depends on $\|f\|_{\mathcal{B}^1}$, $|\Omega|^{\frac{1}{2}}$, $C(c, 0, \text{diam}(\Omega), \sigma)$, and the fixed normalization, but is independent of $w$. Note that $w^{-\frac{1}{2}}$ is convex on $(0, \infty)$. Thus, $U(w)$ is discretely convex in $w$. For $L = 1$ and fixed $q$, the parameter count $d = (2q + 2)w + 1$ is affine in $w$. Thus, the upper bound can be rewritten as a function of $d$:
    \[U(d) = C_ud^{-\frac{1}{2}},\]
    and $U(d)$ inherits monotonicity and discrete convexity from $U(w)$. 
\end{proof}

\begin{theorem}[Siegel and Xu, 2020~\cite{siegel2020nnupperbound}]
\label{thm:siegelandxu}
    Let $\Omega \subset \mathbb{R}^d$ be a bounded domain. If the activation function $\sigma \in W^{m,\infty}(\mathbb{R})$ is non-zero and satisfies the polynomial decay condition
    \begin{equation}
    \label{eq:siegelandxu}
        |\sigma^{(k)}(t)| \leq C_p(1 + |t|)^{-p}
    \end{equation}
    for $0 \leq k \leq m$ and some $p > 1$, we have
    \[\inf_{f_n \in \mathcal{F}_d^n(\sigma)}\|f -f_n\|_{H^m(\Omega)} \leq |\Omega|^{\frac{1}{2}}C(p,m,\diam(\Omega),\sigma)n^{-\frac{1}{2}}\|f\|_{\mathcal{B}^{m+1}},\]
    for any $f \in \mathcal{B}^{m+1}.$
\end{theorem}

\begin{corollary}[Siegel and Xu, 2020~\cite{siegel2020nnupperbound}]
\label{cor:siegelandxu}
    Let $W_{loc}^{m, \infty}(\mathbb{R})$ denote the space of functions whose weak derivatives up to order $m$ are essentially bounded on every compact subset of $\mathbb{R}$. 
    Let $\sigma \in W_{loc}^{m, \infty}(\mathbb{R})$ be an activation function and suppose that there exists a $\nu \in \mathcal{F}_1^{n_0}(\sigma)$ which satisfies the polynomial decay condition~\eqref{eq:siegelandxu} in Theorem~\ref{thm:siegelandxu}, we have
    \[\inf_{f_n \in \mathcal{F}_d^n(\sigma)}\|f -f_n\|_{H^m(\Omega)} \leq |\Omega|^{\frac{1}{2}}C(p,m,\diam(\Omega),\sigma)\sqrt{n_0}\|f\|_{\mathcal{B}^{m+1}}n^{-\frac{1}{2}}.\]
\end{corollary}
This follows immediately from Theorem~\ref{thm:siegelandxu} and the observation that $\nu \in \mathcal{F}_1^{n_0}(\sigma)$ implies that
\[\mathcal{F}_d^{n}(\nu) \subset \mathcal{F}_{d}^{nn_0}(\sigma).\]

\begin{ffnlemma}[FFN admits a lower envelope for Assumption~\ref{A2}]\label{lem:ffn-a2l}
Fix the Fourier feature embedding dimension $q$ and consider one-hidden-layer ReLU FFNs of width $w$. Let $1 \leq p <+\infty$ and $\mathcal{X}$ be a measurable space endowed with a probability measure $\mu$. Let $F$ be a set of measurable functions from $\mathcal{X}$ to $[-1,1]$ such that $\log M(\varepsilon,F,\|\cdot\|_{L^p(\mu)}) \geq c_0\varepsilon^{-\alpha}$ for some $c_0, \varepsilon_0, \alpha > 0$ and all $\varepsilon \leq \varepsilon_0$. Here, $M(\varepsilon,F,\|\cdot\|_{L^p(\mu)})$ denotes the $\epsilon$-packing number of $F$ under $\|\cdot\|_{L^p(\mu)}$, defined as the maximum number of points in $F$ that are pairwise more than $\epsilon$ apart. Then, there exist a positive integer $d_{\min}$ and $c_1 > 0$ such that, for any $d \geq d_{\min}$, the approximation error is bounded below by a monotonically decreasing envelope
\[L(d) = c_1d^{-\frac{1}{\alpha}}\log^{-\frac{3}{\alpha}}(d).\]
\end{ffnlemma}
\begin{proof}
    Define the FFN function class at width $w$ as
    \[\mathcal{G}_w = \{v \to \psi(\gamma_q(v)): \psi\in \Psi_{w,2q}\},\]
    where $\Psi_{w,2q}$ denotes the class of single-hidden-layer, width-$w$ ReLU networks taking inputs in $\mathbb{R}^{2q}$. 
    We first relate the complexity of $\mathcal{G}_w$ to the complexity of $\Psi_{w,2q}$. For a class $S$ of real-valued functions on a domain $Z$, a set of points $\{z_1,\ldots,z_m\} \subseteq Z$ is said to be pseudo-shattered by $S$ if there exists $r_1,\ldots,r_m \in \mathbb{R}$ such that for every sign pattern $o \in \{-1, +1\}^m$ there exists $s \in S$ satisfying $\text{sign}(s(z_i) - r_i) = o_i$ for all $i=1,\ldots,m$. The pseudo-dimension $\text{Pdim}(S)$ is the largest integer $m$ such that some $m$ points are pseudo-shattered by $S$. 
    
    A related notion is the fat-shattering dimension~\cite{achour2022nnlowerbound}. For $\gamma > 0$, a set of points $\{z_1,\ldots,z_m\} \subseteq Z$ is $\gamma$-shattered by $S$ if there exists $r_1,\ldots,r_m \in \mathbb{R}$ such that for every sign pattern $o \in \{-1, +1\}^m$ there exists $s \in S$ satisfying $s(z_i) \geq r_i + \gamma$ when $o_i = +1$ and $s(z_i) \leq r_i - \gamma$ when $o_i=-1$ for $i=1,\ldots,m$. The $\gamma$-fat-shattering dimension $\text{fat}_\gamma(S)$ is the largest integer $m$ such that some $m$ points are $\gamma$-shattered by $S$, with the convention $\text{fat}_\gamma(S) = 0$ if no such set exists and $\text{fat}_\gamma(S) = +\infty$ if arbitrarily large such sets exist. The fat-shattering dimension is dominated by the pseudo-dimension, $\text{fat}_\gamma(S) \leq \text{Pdim}(S)$ for all $\gamma > 0$. Thus, once we show that $\mathcal{G}_w$ has finite pseudo-dimension, the condition $\text{Pdim}(G) < +\infty$ required by Theorem~\ref{thm:achour} follows automatically. 

    We claim 
    \[\text{Pdim}(\mathcal{G}_w) \leq \text{Pdim}(\Psi_{w,2q}).\] 
    Indeed, if a set of points $\{v_i\}_{i=1}^n \subset \mathcal{X}$ is pseudo-shattered by $\mathcal{G}_w$ with $(r_1,\ldots,r_n)$, then for every $o \in \{-1, +1\}^n$, there exists $g \in \mathcal{G}_w$ with $\text{sign}(g(v_i) - r_i) = o_i$ for all $i$. By constructing $g = \psi \circ \gamma_q$ for some $\psi \in \Psi_{w,2q}$, the corresponding embedded points $\{\gamma_q(v_i)\}_{i=1}^n \subset \mathbb{R}^{2q}$ are then pseudo-shattered by $\Psi_{w,2q}$ with the same $(r_1,\ldots,r_n)$, so $n \leq \text{Pdim}(\Psi_{w,2q})$. It follows that $\text{Pdim}(\mathcal{G}_w) \leq \text{Pdim}(\Psi_{w,2q})$. This argument does not require the Fourier embedding to preserve the full complexity of the decoder class; it only uses the fact that composing with a fixed embedding cannot increase pseudo-dimension. 
    
    The number of trainable parameters in $\Psi_{w,2q}$ is $d = (2q+2)w + 1$. By the classical pseudo-dimension bound for ReLU networks with $d$ weights and depth $L=1$~\cite{bartlett2017pseudo}, $\text{Pdim}(\Psi_{w,2q}) = O(d\log d)$, or equivalently, $\text{Pdim}(\Psi_{w,2q}) \leq Cd\log d$ for a universal constant $C$. Combining this with the previous inequality gives 
    \[\text{Pdim}(\mathcal{G}_w) \leq Cd\log d.\]
    Applying Theorem~\ref{thm:achour} to the function class $\mathcal{G}_w$ yields
    \[\sup_{f\in F}\inf_{g \in \mathcal{G}_w}\|f-g\|_{L^p(\mu)} \geq \min \left\{ \varepsilon_1, c_2\text{Pdim}(\mathcal{G}_w)^{-\frac{1}{\alpha}}\log^{-\frac{2}{\alpha}}(\text{Pdim}(\mathcal{G}_w))\right\},\]
    for some constants $\epsilon_1, c_2 > 0$. Let $d_{\min}$ be a positive integer such that for any $d > d_{\min}$, the second term is smaller than $\epsilon_1$. Since $t^{-\frac{1}{\alpha}}\log^{-\frac{2}{\alpha}}(t)$ is decreasing in $t$, we can substitute the bound $\text{Pdim}(\mathcal{G}_w) \leq Cd\log d$ to obtain
    \[\sup_{f\in F}\inf_{g \in \mathcal{G}_w}\|f-g\|_{L^p(\mu)} \geq c_2(Cd\log d)^{-\frac{1}{\alpha}}\log^{-\frac{2}{\alpha}}(Cd\log d).\]
    Absorbing all logarithmic constants $c_1$ yields
    \[\sup_{f\in F}\inf_{g \in \mathcal{G}_w}\|f-g\|_{L^p(\mu)} \geq c_1d^{-\frac{1}{\alpha}}\log^{-\frac{3}{\alpha}}(d).\]
    In other words, $L(d) = c_1d^{-\frac{1}{\alpha}}\log^{-\frac{3}{\alpha}}(d)$ is the desired lower envelope for the global representation error over the target class $F$. 
\end{proof}

\begin{theorem}[Achour, 2022~\cite{achour2022nnlowerbound}]
\label{thm:achour}
    Let $1 \leq p <+\infty$ and $\mathcal{X}$ be a measurable space endowed with a probability measure $\mu$. Let $F, G$ be two sets of measurable functions from $\mathcal{X}$ to $\mathbb{R}$, such that all functions in $F$ have the same $[a,b]$ for some $a < b$, and such that $\text{fat}_\gamma(G) < +\infty$ for all $\gamma > 0$. If $\log M(\varepsilon,F,\|\cdot\|_{L^p(\mu)}) \geq c_0\varepsilon^{-\alpha}$ for some $c_0, \varepsilon_0, \alpha > 0$ and all $\varepsilon \leq \varepsilon_0$, and if $\text{Pdim}(G) < +\infty$, then there exist constants $c_1, \varepsilon_1 > 0$ depending only on $b-a, p, c_0, \varepsilon_0$ and $\alpha$ such that 
    \[\sup_{f\in F}\inf_{g \in G}\|f-g\|_{L^p(\mu)} \geq \min \left\{ \varepsilon_1, c_1\text{Pdim}(G)^{-\frac{1}{\alpha}}\log^{-\frac{2}{\alpha}}(\text{Pdim}(G))\right\}.\]
    Here $\text{Pdim}(G)$ denotes the pseudo-dimension of $G$. 
\end{theorem}

\begin{ffnlemma}[FFN satisfies Assumption~\ref{A3}]\label{lem:ffn-a3}
There exists a finite parameter budget $d^\star$ such that $\|x - g_{d^\star}^\star\| = 0$ with high probability.
\end{ffnlemma}
\begin{proof}
    The proof proceeds in two steps. First, we show that there exists a finite embedding dimension $q^\star$ such that the Fourier feature matrix has full row rank, guaranteeing exact interpolation of the target signal at all sample points. Second, we invoke the minimum-width universal approximation results to show that a single-hidden-layer ReLU MLP of finite width can realize the required linear map on the embedded coordinates.

    \textit{Step 1: Full row rank of the Fourier feature matrix.} Let $S = \{v_1, v_2, \ldots, v_n\}$ be the finite set of distinct sampled coordinates at which the signal $x$ is defined, and form the Fourier feature matrix 
    \[
    \Phi_S := \begin{bmatrix}
    \gamma_q(v_1)^\top \\
    \vdots \\
    \gamma_q(v_n)^\top
    \end{bmatrix} \in \mathbb{R}^{n\times 2q}.
    \]
    We claim that there exists $q^\star$ with $2q^\star \geq n$ such that $\rank(\Phi_S) = n$ with high probability over the random draws of the embedding frequencies. To establish this, factor $\Phi_S = \Phi_S' \diag(a_1, a_1, \ldots, a_q, a_q)$ where
    \[\Phi_S' = \begin{bmatrix}
\cos(2\pi b_1^\top v_1),\;\sin(2\pi b_1^\top v_1)\; & \cdots & \;
\cos(2\pi b_q^\top v_1),\;\sin(2\pi b_q^\top v_1) \\
\vdots & \ddots & \vdots\\
\cos(2\pi b_1^\top v_n),\;\sin(2\pi b_1^\top v_n)\; & \cdots & \;
\cos(2\pi b_q^\top v_n),\;\sin(2\pi b_q^\top v_n)
\end{bmatrix}.\]
    Since $a_l \sim \mathcal{N}(0,\sigma_a^2)$ for $l = 1,\ldots,q$, all $2q$ entries on the diagonal are nonzero almost surely, so the diagonal matrix $\diag(a_1, a_1, \ldots, a_q, a_q)$ has full rank with probability one. Hence, $\rank(\Phi_S) = \rank(\Phi_S')$ with probability one and it suffices to show that $\Phi_S'$ has full row rank with high probability. Consider the Gram matrix $G = \Phi'_S (\Phi'_S)^\top = \sum_{l=1}^{q} G_l$, where each $G_l \in \mathbb{R}^{n\times n}$ is the symmetric positive semidefinite matrix with entries $G_l[i,j] = \cos(2\pi b_l^\top (v_i - v_j))$. We note that $G_l = U_l U_l^\top$ where 
    \[U_l = \begin{bmatrix}
    \cos(2\pi b_l^\top v_1) & \sin(2\pi b_l^\top v_1) \\
    \cos(2\pi b_l^\top v_2) & \sin(2\pi b_l^\top v_2) \\
    \vdots & \vdots \\
    \cos(2\pi b_l^\top v_n) & \sin(2\pi b_l^\top v_n) \\
    \end{bmatrix} \in \mathbb{R}^{n\times 2}.\]
    It follows that $G_l \succeq 0$ and $\|G_l\| \leq \|U_l\|^2 \leq \|U_l\|_F^2 = tr(U_l^\top U_l) = n$. Let $\Sigma = \mathbb{E}[G_l]$.
    Then the $(i,j)$ entry of $\Sigma$ is
    \[\Sigma[i,j] = \mathbb{E}[\cos(2\pi b_l^\top (v_i - v_j))] = \exp\left(-2\pi^2 \sigma_b^2 \|v_i-v_j\|_2^2\right), \]
    which is the Gaussian RBF kernel evaluated at $v_i - v_j$. Since the points $v_1,\ldots,v_n$ are distinct, $\Sigma$ is a Gram matrix of a Gaussian RBF kernel on distinct points and is therefore strictly positive definite. Let $\lambda = \lambda_{\min}(\Sigma) > 0$ be the smallest eigenvalue of $\Sigma$.

    Define the centered, rescaled matrices $W_l = \Sigma^{-1/2}G_l\Sigma^{-1/2}$. Then $W_l$ is symmetric and positive semidefinite by construction. We also have that 
    \[\mathbb{E}[W_l] = \Sigma^{-1/2}\mathbb{E}[G_l]\Sigma^{-1/2} = \Sigma^{-1/2}\Sigma \Sigma^{-1/2}=I_n, \qquad \text{and}\]
    \[
    \|W_l\| \leq \|\Sigma^{-1/2}\|\|G_l\|\|\Sigma^{-1/2}\|
    = \|\Sigma^{-1/2}\|^2\|G_l\| 
    = \left(\frac{1}{\sqrt{\lambda_{\min}(\Sigma)}}\right)^2\|G_l\| \leq \frac{n}{\lambda} .\]
    Let the diagonalization of $W_l$ be $W_l = PDP^\top$, where $P$ is orthogonal and $D = \diag(\lambda_1, \ldots, \lambda_n)$. By definition, for all $i \in [n]$, $\lambda_i \leq \lambda_{\max}(W_l) = \|W_l\|$ and $\lambda_i^2 \leq \|W_l\|\lambda_i$. It follows that $D^2 \preceq \|W_l\|D$ and $W_l^2 = PD^2P^\top \preceq \|W_l\| PDP^\top = \|W_l\|W_l$. Define $G'_l = \frac{1}{q}(W_l - I)$. We have
    \[\|G'_l\| \leq \frac{1}{q}(\|W_l\| + \|I_n\|) \leq \frac{1}{q}\left(\frac{n}{\lambda} + 1\right),\]
    \[\mathbb{E}[G'_l] =  \frac{1}{q}(\mathbb{E}[W_l] - I_n) = 0_n, \qquad \text{and}\]
    \begin{align*}
        \mathbb{E}[(G'_l)^2] &= \frac{1}{q^2}\mathbb{E}[W_l^2 - 2W_l + I_n]\\
        &= \frac{1}{q^2}\mathbb{E}[W_l^2] - 2I_n + I_n \\
        &\preceq \frac{1}{q^2} \left(\|W_l\|\mathbb{E}[W_l] - I_n\right) \\
        &\preceq \frac{1}{q^2}\left(\frac{n}{\lambda} - 1\right)I_n \preceq \frac{n}{q^2\lambda}I_n.
    \end{align*}
    Thus, $G'_l$ are independent centered self-adjoint random $n \times n$ matrices with
    \[\|G'_l\| \leq \frac{1}{q}\left(\frac{n}{\lambda} + 1\right) \qquad \text{ and } \qquad \left\|\sum_{l=1}^q \mathbb{E}[(G'_l)^2]\right\| \leq \frac{n}{q\lambda}.\]
    Applying the non-commutative Bernstein inequality (Theorem~\ref{thm:vershynin}) to the $G'_l$ matrices, we have that for every $\epsilon\geq 0$,
    \[
        \delta = \mathbb{P}\left\{ \|\sum_{l=1}^q G'_l\| \geq \epsilon \right\} \leq 2n \cdot \exp\left(\frac{-\epsilon^2/2}{\frac{n}{q\lambda} + \frac{(n/\lambda + 1)\epsilon}{3q}}\right) = 2n\cdot \exp\left(\frac{-\epsilon^2q}{\frac{n}{\lambda}\cdot\frac{2\epsilon +6}{3} + \frac{2\epsilon}{3}}\right).
    \]
    Pick $\epsilon = 1/2$. We get
    \[
         \mathbb{P}\left\{ \|\sum_{l=1}^q G'_l\| < \frac{1}{2} \right\} = 1 - \delta  \geq 1 - 2n\cdot \exp\left(\frac{-q}{\frac{n}{\lambda}\cdot\frac{28}{3} + \frac{4}{3}}\right).
    \]
    This implies that if $q \geq C\frac{n}{\lambda}\ln\frac{2n}{\delta}$ for some absolute constant $C$, then $\mathbb{P}\left\{ \|\sum_{l=1}^q G'_l\| < \frac{1}{2} \right\} \geq 1 - \delta$. We can write
    \[
        \|\sum_{l=1}^q G'_l\| < \frac{1}{2} \iff \|\sum_{l=1}^q \left(\frac{1}{q}W_l - \frac{1}{q}I_n\right)\| < \frac{1}{2} \iff \|\left(\frac{1}{q}\sum_{l=1}^q W_l \right) - I_n\| < \frac{1}{2}. 
    \]
    Since $W_l$ is symmetric, $\left(\frac{1}{q}\sum_{l=1}^q W_l \right) - I_n$ is also symmetric and $\left|\lambda_{\max}\left(\left(\frac{1}{q}\sum_{l=1}^q W_l \right) - I_n\right)\right| = \|\left(\frac{1}{q}\sum_{l=1}^q W_l \right) - I_n\|$. It follows that $ \left|\lambda_{\max}\left(\left(\frac{1}{q}\sum_{l=1}^q W_l \right) - I_n\right)\right| < \frac{1}{2}$ and all eigenvalues of $\left(\frac{1}{q}\sum_{l=1}^q W_l \right) - I_n$ lie in $(-\frac{1}{2},\frac{1}{2})$. As a result, all eigenvalues  $\frac{1}{q}\sum_{l=1}^q W_l$ lies in $(\frac{1}{2},\frac{3}{2})$. It follows that 
    \begin{align*}
        \frac{1}{q}\sum_{l=1}^q W_l \succeq \frac{1}{2}I_n &\iff \Sigma^{-1/2}\left(\sum_{l=1}^q G_l\right)\Sigma^{-1/2} \succeq \frac{q}{2}I_n\\
        &\iff G \succeq \frac{q}{2}\Sigma \\
        &\iff \lambda_{\min}(G) \geq \frac{q}{2}\lambda > 0.
    \end{align*}
    Therefore, we can conclude that for $q \geq C\frac{n}{\lambda}\ln\frac{2n}{\delta}$ for some absolute constant $C$, then $\rank(\Phi_S') = \rank(G) =n$ with probability at least $1 - \delta$.
    
    Since $\Phi_S$ has full row rank, there exists a vector $c \in \mathbb{R}^{2q^\star}$ such that $\Phi_S c = x_S$ (for instance, the minimum-norm solution $c = \Phi_S^\top (\Phi_S \Phi_S^\top)^{-1}x_S)$). The linear map $v \to c^\top \gamma_{q^\star}(v)$ therefore interpolates $x$ exactly at all $n$ sampled points.
    
    \textit{Step 2: Realization by a finite-width ReLU MLP.} It remains to show that the linear map $z \to c^\top z$ can be realized by a single-hidden-layer ReLU MLP of finite width. This is immediate: a width-$4q^\star$ network with weights
    \[
        W_1 = \begin{bmatrix}
            I_{2q^\star} \\
            -I_{2q^\star}
        \end{bmatrix} \in \mathbb{R}^{4q^\star \times 2q^\star}, \;\;\; b_1 = 0 \in \mathbb{R}^{4q^\star}, \;\;\; W_{out} = \begin{bmatrix}
            c^\top \;\; -c^\top
        \end{bmatrix} \in \mathbb{R}^{1 \times 4q^\star}, \;\;\; b_{out} = 0 \in \mathbb{R}
    \]
    satisfies $W_{out}\cdot \ReLU(W_1z) = c^\top(\ReLU(z) - \ReLU(-z)) = c^\top z$ for all $z$, since $\ReLU(z) - \ReLU(-z) = z$. In particular, evaluating at $z = \gamma_{q^\star}(v_i)$ yields $c^\top \gamma_{q^\star}(v_i) = x_{S, i}$ for all $i =1, \ldots,n$. 

    More generally, for the effective input dimension $d_x = 2q$ and output dimension $d_y = 1$,the minimum width for $L^p$-approximation of any function in $L^p([0,1]^{d_x},\mathbb{R}^{d_y})$ by ReLU networks with at least 2 hidden layers is $w_{\min} = \max\{d_x, d_y, 2\}$, as established by Kim et al.~\cite{kim2024tighterminimumwidth}. This guarantees that any continuous function on the compact domain of the Fourier embedding, including the above interpolant, can be approximated to arbitrary precision by a ReLU MLP of at least 2 hidden layers and width at least $2q$. The critical parameter budget is therefore either $d^\star = P(L=1, w=4q^\star, q=q^\star)$ or $d^\star = P(L=2, w=2q^\star, q=q^\star)$, which are both finite. 
\end{proof}
\begin{theorem}[Non-commutative Bernstein-type inequality (Vershynin, 2011~\cite{vershynin2011randommatrices})]
\label{thm:vershynin}
    Consider a finite sequence $X_i$ of independent centered self-adjoint random $n\times n$ matrices. Assume we have for some numbers $K$ and $\sigma$ that
    \[\|X_i\| \leq K \;\;\text{ almost surely,} \qquad \|\sum_i\mathbb{E}X_i^2\| \leq \sigma^2.\]
    Then, for every $t \geq 0$ we have
    \[\mathbb{P}\left\{\|\sum_iX_i\| \geq t\right\} \leq 2n\cdot\exp \left(\frac{-t^2/2}{\sigma^2 + Kt/3}\right).\]
\end{theorem}

\subsection{Proof of Theorem~\ref{thm:model-specific}: Instant-NGP}
We consider an Instant-NGP \cite{muller2022ingp} model with $L$ resolution levels, $T$ entries per level hash table, $F$ feature dimensions per entry, and a single-hidden-layer ReLU MLP decoder of width $w$. The number of trainable encoding parameters is $TLF$, and the total parameter count $d = TLF + d_{MLP}$ includes both encoding and decoder parameters. Define the full concatenated encoding at coordinate $v \in \mathbb{R}^k$ as
\[z(v) = \begin{bmatrix} e_1(v) & e_2(v) & \cdots & e_L(v) & \xi(v)\end{bmatrix} \in \mathbb{R}^{LF + E},\]
where $e_l(v)$ is the interpolated hash-grid feature at level $l$ and $\xi(v) \in \mathbb{R}^E$ denotes an auxiliary non-spatial encoding. Following~\cite{muller2022ingp}, we index the model family by the encoding hyperparameters $(T,L)$ while holding $F$, the spatial input dimension $k$, the coarsest resolution $N_{\min}$, and the decoder architecture fixed. For a target signal $x \in \mathbb{R}^n$ sampled at coordinates $v_1,\ldots,v_n \in \mathbb{R}^k$, the optimal reconstruction error at parameter budget $d$ is 
\[\|x-g_d^\star\| = \min_{\theta, \Phi}\|x - g_d(\theta, \Phi)\|,\]
where $\theta$ denotes the hash table feature parameters and $\Phi$ encompasses the MLP weights and biases.

To analyze the error curve, we adopt two standard structural assumptions on the spatial hash function used in the multi-resolution encoding~\cite{teschner2003hash}: 
\begin{itemize}[leftmargin=10pt]
    \item \textbf{Uniform Hashing.} Let $\hat{h}_T$ denote the spatial hash function used by the multi-resolution encoding. Then
    $\Pr[\hat{h}_T(u) = \hat{h}_T(v)] \leq \frac{1}{T}$ for any two distinct integer vertices $u \neq v$. 
    \item \textbf{Independence Across Levels.} Hash collision events at different resolution levels are mutually independent. 
\end{itemize}
These assumptions reflect the design objectives of the per-dimension linear-congruential XOR hash used in Instant-NGP, namely decorrelating the contribution of the input dimensions and approximately uniform distribution over the table.

\begin{ingplemma}[Instant-NGP satisfies Assumption~\ref{A1}]
\label{lem:ingp-a1}
    For all parameter budgets $d \geq 0$,
    \[\|x - g_d^\star\| \geq \|x - g_{d+1}^\star\|.\]
\end{ingplemma}
\begin{proof}
    It suffices to show that the model class $\mathcal{G}_d$ is contained in $\mathcal{G}_{d'}$ whenever $d' > d$. We first consider increasing the number of levels. A model with $L$ levels computes an encoding of the form
    \[z(v) = \begin{bmatrix} e_1(v) & e_2(v) & \cdots & e_L(v) & \xi(v)\end{bmatrix} \in \mathbb{R}^{LF + E},\]
    where $e_l(v)$ is the interpolated hash-grid feature at level $l$ and $\xi(v) \in \mathbb{R}^E$ denotes any auxiliary non-spatial encoding. Let the first affine layer of the decoder have weights 
    \[W = \begin{bmatrix} W_{\text{hash}} &  W_{\xi}\end{bmatrix},\]
    where $W_{\text{hash}}$ acts on the $LF$-dimensional concatenated hash features. A model with $L + 1$ levels has encoding
    \[z(v) = \begin{bmatrix} e_1(v) & e_2(v) & \cdots & e_L(v) & e_{L+1}(v) &\xi(v)\end{bmatrix} \in \mathbb{R}^{(L+1)F + E}.\]
    It can reproduce the original $L$-level model by choosing the new first-layer weights as
    \[W = \begin{bmatrix} W_{\text{hash}} & 0 &  W_{\xi}\end{bmatrix},\]
    setting the new incoming weights from level $L+1$ to zero and keeping all remaining decoder weights unchanged. The additional level may contain arbitrary feature values, but it contributes nothing to the output. Thus every reconstruction realizable with $L$ levels is also realizable with $L+1$ levels. 

    We now consider increasing the hash table size by an integer multiple. Let $T' = mT$ for an integer $m \geq 1$. At a fixed level, the spatial hash function has the form
    \[\hat{h}_T(p) = H(p) \mod T,\]
    where $p$ is an integer grid vertex and $H(p) = \bigoplus_{i=1}^k p_i \pi_i$  is the pre-modulo hash value with $\pi_i$ being unique, large prime numbers. Given any table $\theta^{(l)} \in \mathbb{R}^{T\cdot F}$ for the smaller model, define the larger table features $\theta'^{(l)} \in \mathbb{R}^{T'\cdot F}$ by 
    \[\theta'^{(l)}_j = \theta^{(l)}_{j \mod T}\]
    Since $T' = mT$, we have $\hat{h}_{T'}(p) \mod T = \hat{h}_T(p)$. Then, for any level $l$ and corner $p$, the looked-up feature in the larger model $\theta'^{(l)}_{\hat{h}_{T'}(p)} = \theta^{(l)}_{\hat{h}_{T'}(p) \mod T} = \theta^{(l)}_{\hat{h}_{T}(p)}$ matches the corresponding lookup in the smaller model. With this construction of the larger model, the interpolation weights depend only on the input coordinate and grid resolution, not on $T$, so the interpolated feature at every level is identical in the two models. Keeping the weights fixed gives the same output for every input coordinate $v_i$ for $i=1,\ldots,n$. Thus, $\mathcal{G}_{T} \subseteq \mathcal{G}_{mT}$. 

    In both cases (increasing the number of levels $L$ or hash table size $T$), the model class is nested, so the infimum of reconstruction error over the larger class is no greater than over the smaller class. Therefore, $\|x - g_d^\star\|$ is nonincreasing in $d$.
\end{proof}

\begin{ingplemma}[For Assumption~\ref{A2}, Instant-NGP admits a monotone, discretely convex upper envelope for the expected optimal squared error, where the expectation is taken over the random hash function]
\label{lem:ingp-a2u}
    Under the uniform hashing and level-independence assumptions, the expected squared optimal reconstruction error of Instant-NGP, taken over the random spatial hash function, satisfies
    \[\mathbb{E}[\|x-g_d^\star\|^2] \leq C_n \left( \frac{C_k}{T}\right)^L = U(d)^2,\]
    where $C_n = \frac{1}{2} \binom{n}{2}, C_k = 4^k$, and $d^\star$ is the critical parameter budget identified in~\Cref{lem:ingp-a3} (below). The upper envelope $U(d)$ is monotone decreasing and discretely convex in $T$ (with $L,F$ fixed) and in $L$ (with $T,F$ fixed). 
\end{ingplemma}
\begin{proof}
    Let $U_l(v_i)$ be the set of $2^k$ integer grid vertices involved in the $k$-linear interpolation of $v_i$ at level $l$. For the level-$l$ interpolated feature vectors for $v_i$ and $v_j$ to be indistinguishable, at least one of the queried vertices for $v_i$ must collide with one of the queried vertices for $v_j$. By union bound and the uniform hashing assumption, 
    \begin{align*}
    \Pr[A_{ij}^{(l)}] &\leq \Pr \left[ \bigcup_{u_i \in U_l(v_i)} \bigcup_{u_j \in U_l(v_j)} \{h(u_i) = h(u_j)\}\right] \\
    &\leq \sum_{u_i \in U_l(v_i)} \sum_{u_j \in U_l(v_j)} \Pr[h(u_i) = h(u_j)] \\
    &\leq |U_l(v_i)|\cdot|U_l(v_j)|\cdot \frac{1}{T} \\
    &\leq 2^k \cdot 2^k \cdot \frac{1}{T} = \frac{C_k}{T},
    \end{align*}
    where $A_{ij}^{(l)}$ is the event that the level-$l$ encoded features for $v_i$ and $v_j$ collide, and $C_k = 4^k$. Since the full encoding concatenates all $L$ levels, the two samples are indistinguishable to the decoder only if they are indistinguishable at every level. By independence across levels,
    \[\Pr[A_{ij}] = \prod_{l =1}^L \Pr[A_{ij}^{(l)}] \leq \left( \frac{C_k}{T}\right)^L,\]
    where $A_{ij} = \bigcap_{l=1}^{L} A_{ij}^{(l)}$ is the event that the interpolated feature vectors $v_i$ and $v_j$ are indistinguishable at all levels. Whenever $A_{ij}$ occurs, the decoder receives the same input for $v_i$ and $v_j$, and therefore every Instant-NGP reconstruction must assign the same predicted value to these samples. More generally, the relation of structural indistinguishability partitions $\{1,\ldots,n\}$ into equivalence classes $C \in \Pi_{T,L}$, and every reconstruction is constant on each such class. The optimal MLP outputs the within-class mean $\bar{x}_{C}$ on each class $C$, so
    \begin{align*}
        \|x - g_{T,L}^\star\|^2 &=  \sum_{C \in \Pi_{T,L}} \sum_{i \in C} (x_i - \bar{x}_C)^2 \\
        &= \sum_{C \in \Pi_{T,L}} \sum_{i \in C} (x_i - \frac{1}{|C|}\sum_{j \in C}x_j)^2 \\
        &= \sum_{C \in \Pi_{T,L}} \frac{1}{|C|}\sum_{i,j \in C}(x_i - x_j)^2 \\
        &= \sum_{C \in \Pi_{T,L}} \frac{1}{2|C|}\sum_{\substack{i,j \in C \\ i<j}}(x_i - x_j)^2.
    \end{align*}
    Using $1/|C| \leq 1 \; \forall C \in \Pi_{T,L}$ and the normalization $x_i \in [0,1]$, we can write
    \[\|x - g_{T,L}^\star\|^2 
        \leq \frac{1}{2}\sum_{C \in \Pi_{T,L}} \sum_{\substack{i,j \in C \\ i<j}}(x_i - x_j)^2 \leq \frac{1}{2}\sum_{i<j}\mathds{1}[A_{ij}].\]
    Taking expectation over the random hash function gives the upper envelope:
    \begin{align*}
        \mathbb{E}[\|x - g_{T,L}^\star\|^2] &\leq \frac{1}{2}\sum_{i<j}\Pr[A_{ij}] \\
        &\leq \frac{1}{2} \binom{n}{2} \left( \frac{C_k}{T}\right)^L, \;\;\; \text{ since there are $\binom{n}{2}$ pairs of input coordinates} \\
        &= C_n \left( \frac{C_k}{T}\right)^L = U(T, L)^2, \;\;\; \text{ where $C_n$ is a constant depending on $n$.}
    \end{align*}
    It remains to verify the monotone discrete convexity of the upper envelope. For fixed $L$, 
    \[U(T,L) = CT^{-L/2}\]
    for a constant $C > 0$. The function $f(t) = t^{-L/2}$ is decreasing and convex on $(0,\infty)$, for $L\geq 1$, since its first derivative is negative and its second derivative is positive. Therefore, its restriction to integer table sizes is discretely convex:
    \[U(T-1, L) - U(T,L) \geq U(T,L) - U(T+1,L).\]
    For fixed $T > C_k$, let $q = (C_k/T)^{1/2} \in (0,1)$. Then 
    \[U(T,L) = Cq^L\]
    for a constant $C>0$. Hence
    \[U(T, L-1) - U(T, L) - (U(T, L) - U(T, L+1)) = C(q^{L-1} - 2q^L + q^{L+1}) = Cq^{L-1}(q-1)^2 > 0.\]
    Thus the upper envelope is monotone decreasing and discretely convex in either $T$ or $L$. 
\end{proof}

\begin{ingplemma}[For Assumption~\ref{A2}, Instant-NGP admits a lower envelope for the expected optimal squared error, where the expectation is taken over the random hash function]
\label{lem:ingp-a2l}
    Under the uniform hashing and level-independence assumptions, the expected squared optimal reconstruction error of Instant-NGP, taken over the random spatial hash function, satisfies
    \[\mathbb{E}[\|x-g_d^\star\|^2] \geq \delta_x\cdot\mathds{1}\{d \leq d^\star\} \cdot T^{-2^kL} = L(d)^2  ,\]
    where $\delta_x >0$ is a positive signal-dependent constant and $d^\star$ is the critical parameter budget identified in~\Cref{lem:ingp-a3} (below). The lower envelope $L(d)$ need not be convex.
\end{ingplemma}
\begin{proof}
    In contrast to the upper-bound argument in Lemma~\ref{lem:ingp-a2u}, which uses the necessary condition that indistinguishability at level $l$ requires at least one collision among the queried corner vertices, the lower bound needs the stricter sufficient event that all corresponding corner vertices collide, which guarantees that the interpolated features are identical. This restriction is what allows the upper envelope and the lower bound to be expressed in compatible terms. Let $A_{ij}^{(l)}$ be the event that the level-$l$ encoded features for $v_i$ and $v_j$ collide. Consider any pair $(v_i, v_j)$ that lie at the same relative position within distinct voxels at every level $l$, so that the $k$-linear interpolation weights match. For such a pair, the event $A_{ij}^{(l)}$ reduces to the requirement that each of the $2^k$ corners of $v_i$'s voxel hashes to the same bucket as the corresponding corner of $v_j$'s voxel. Treating the $2^k$ corner pairs as independent uniform draws,
    \[\Pr[A_{ij}^{(l)}] \geq \left(\frac{1}{T}\right)^{2^k},\]
    and by independence across levels,
    \[\Pr[A_{ij}] \geq \left(\frac{1}{T}\right)^{2^kL},\]
    where $A_{ij} = \bigcap_{l=1}^{L} A_{ij}^{(l)}$ is the event that the interpolated feature vectors $v_i$ and $v_j$ are indistinguishable at all levels. Whenever $A_{ij}$ occurs, the decoder receives the same input for $v_i$ and $v_j$, and therefore every Instant-NGP reconstruction must assign the same predicted value to these samples. More generally, the relation of structural indistinguishability partitions $\{1,\ldots,n\}$ into equivalence classes $C \in \Pi_{T,L}$, and every reconstruction is constant on each such class. The optimal MLP outputs the within-class mean $\bar{x}_{C}$ on each class $C$, so
    \begin{align*}
        \|x - g_{T,L}^\star\|^2 &=  \sum_{C \in \Pi_{T,L}} \sum_{i \in C} (x_i - \bar{x}_C)^2 \\
        &= \sum_{C \in \Pi_{T,L}} \sum_{i \in C} (x_i - \frac{1}{|C|}\sum_{j \in C}x_j)^2 \\
        &= \sum_{C \in \Pi_{T,L}} \frac{1}{|C|}\sum_{i,j \in C}(x_i - x_j)^2 \\
        &= \sum_{C \in \Pi_{T,L}} \frac{1}{2|C|}\sum_{\substack{i,j \in C \\ i<j}}(x_i - x_j)^2.
    \end{align*}
    Using $1/|C| \geq 1/n \; \forall C \in \Pi_{T,L}$, we can write
    \[\|x - g_{T,L}^\star\|^2 \geq \frac{1}{2n}\sum_{i<j}\mathds{1}[A_{ij}](x_i - x_j)^2.\]
    Let $\mathcal{P}$ be the set of pairs $(i,j)$ at matching relative positions within their voxels and let $\delta_{\min} = \min_{(i,j) \in \mathcal{P}} (x_i - x_j)^2 > 0$ for a target signal whose sampled values are pairwise distinct on $\mathcal{P}$. Taking expectation over the random hash function yields
    \[\mathbb{E}[\|x - g_{T,L}^\star\|^2] \geq \frac{\delta_{\min}|\mathcal{P}|}{2n}\cdot T^{-2^kL} = \delta_x \cdot T^{-2^kL},\]
    where $\delta_x$ absorbs all signal-dependent constants. When $d \geq d^\star$, Lemma~\ref{lem:ingp-a3} gives $\|x - g_{T^\star,L^\star}^\star\|^2 = 0$, so we set $L(d) = 0$ in this regime. Combining the two regimes yields a valid lower envelope 
    \[L(d) = \left(\delta_x\cdot\mathds{1}\{d \leq d^\star\} \cdot T^{-2^kL}\right)^{1/2}.\]
\end{proof}

\begin{ingplemma}[Instant-NGP satisfies Assumption~\ref{A3}]
\label{lem:ingp-a3}
    There exists a finite parameter budget $d^\star$ such that $\|x - g_{d^\star}^\star\| = 0$. 
\end{ingplemma}
\begin{proof}
    Since the target signal is observed on finitely many sample coordinates $v_1,\ldots,v_n,$ choose the number of levels $L^\star$ such that finest grid resolution $N_{\max}$ is large enough for each sample coordinate to be represented by a distinct finest-level grid vertex. For grid-sampled images, volume, or radiance-field data, it suffices to take $N_{\max} \geq R$, where $R$ denotes the resolution of $x$ along each spatial axis.

    Next, choose the finest-level hash table size so that $T^\star \geq N_{\max}^k$. Then every finest-level grid vertex can be assigned a unique table entry, so the finest level behaves as a dense feature grid with no hash collisions. That is, each input coordinate $v_i$ is assigned a unique feature vector $\theta^{(L)}_i$ at level $L$, and the encoding is injective: $z(v_i) \neq z(v_j)$ for all $i \neq j$. 

    We now construct an exact representation. Set all decoder weights connected to coarser levels to zero and let the decoder depend only on the finest-level feature block. Assign the finest-level feature vector at the vertex corresponding to $v_i$ so that a linear readout returns $x_i$. For scalar-valued signals, one feature channel is sufficient; concretely, 
    set 
    \[\text{MLP}_\Phi(z(v_i)) = \text{MLP}_\Phi\left(\begin{bmatrix}
        \theta_i^{(0)} &\cdots &\theta_i^{(l)} &\cdots &\theta_i^{(L)}
    \end{bmatrix}\right) = \text{MLP}_\Phi\left(\begin{bmatrix}
        0 &\cdots &0 &\theta_i^{(L)}
    \end{bmatrix}\right) = x_i.\]
    For vector-valued signals, one may either use one feature channel per output coordinate or a finite-width linear output layer. Since the Instant-NGP decoder ends in a linear layer, this readout can be implemented by the MLP by choosing the hidden layers to act as an identity map on the relevant feature coordinates, or equivalently by using the final affine layer to read out the stored value directly. With this construction,
    $g_{d^\star}(v_i) = x_i$
    for every $i=1,\ldots,n$. Hence, the optimal error at the corresponding finite budget $d^\star = T^{\star}L^{\star} F + d_{MLP}$ is zero.
\end{proof}

\subsection{Proof of Theorem~\ref{thm:model-specific}: GA-Planes}
\label{appendix:proofs:gap}
We prove the GA-Planes \cite{sivgin2024gaplanes} case for the two-dimensional nonconvex variant introduced in~\cite{sivgin2024gaplanes} and index the model family by its line feature dimension $k_1$, holding the line resolution $r_1$, plane resolution $r_2$, and plane feature dimension $k_2 = 1$ fixed. Let the target signal $x \in \mathbb{R}^n$ be arranged as a matrix $M \in \mathbb{R}^{N \times N}$ with $n = N^2$, and write the model output at integer coordinate $(i,j)$ as
\[g_d(\theta)[i,j] = W_{\text{line}}^\top(e_1[i] \circ e_2[j]) + W_{\text{plane}}^\top e_{12}[i,j,:] + b,\] 
where $e_1, e_2$ are linearly interpolated from line feature grids $g_1, g_2 \in \mathbb{R}^{r_1\times k_1}$, $e_{12}$ is bilinearly interpolated from a plane feature grid $g_{12}\in\mathbb{R}^{r_2\times r_2\times k_2}$, $\circ$ denotes elementwise multiplication along the feature dimension, and $\theta = (g_1, g_2, g_{12}, W_{\text{line}}, W_{\text{plane}}, b)$ are the trainable parameters. 

A central tool in the analysis is the matrix-completion equivalence established in~\cite{sivgin2024gaplanes} for the 2D linear-decoder setting: the reconstructed matrix $\hat{M} \in \mathbb{R}^{N\times N}$ produced by $g_d$ admits the decomposition
\[\hat{M} = UV^T + \varphi(L),\]
with $U = g_1 \diag(W_{\text{line}}) \in \mathbb{R}^{N \times k_1}, V=g_2 \in \mathbb{R}^{N\times k_1}, L\in \mathbb{R}{r_2}{r_2}$ encoding the rescaled and bias-shifted plane grid, and $\varphi:\mathbb{R}^{r_2 \times r_2} \to \mathbb{R}^{N\times N}$ the bilinear upsampling operator. The optimal squared compression loss at line feature dimension $k_1$ therefore equals
\begin{equation}
\label{matrix-model}
    \|x - g_{k_1}^\star\|^2 = \min_{U,V,L} \| M - (UV^T + \varphi(L)) \|_F^2.
\end{equation}
This casts the analysis of 2D GA-Planes as a low-rank plus low-resolution matrix completion problem. 

\begin{gaplemma}[GA-Planes satisfies Assumption~\ref{A1}]
\label{lem:gap-a1}
For all $k_1 \geq 0$,
\[\|x - g_{k_1}^\star\| \geq \|x - g_{k_1 + 1}^\star\|.\]
\end{gaplemma}
\begin{proof}
    Let $\mathcal{G}_{k_1}$ denote the set of model outputs realizable at line feature dimension $k_1$. It suffices to show $\mathcal{G}_{k_1} \subseteq \mathcal{G}_{k_1 + 1}$. Let $g \in \mathcal{G}_{k_1}$  be represented by the parameters $(g_1, g_2, g_{12}, W_{\text{line}}, W_{\text{plane}}, b)$. Construct a model in $\mathcal{G}_{k_1 + 1}$ by zero-padding the line components
    \[g_1' = [g_1, 0] \in \mathbb{R}^{r_1\times(k_1 + 1)}, \;\; g_2' = [g_2, 0] \in \mathbb{R}^{r_1\times(k_1 + 1)}, \;\; W'_{\text{line}} = \begin{bmatrix} W_{\text{line}} \\ 0 \end{bmatrix}\in \mathbb{R}^{k_1 + 1}\]
    while keeping the plane and bias components unchanged: $g'_{12} = g_{12}, W'_{\text{plane}} = W_{\text{plane}}, b'=b$. Since linear interpolation is applied channelwise, the new interpolated line features satisfy 
    \[e'_1 = \begin{bmatrix} e_1 \\ 0
    \end{bmatrix}, \;\;\; e'_2 = \begin{bmatrix} e_2 \\ 0
    \end{bmatrix}, \;\;\; \text{and hence}\;\;\; (e'_1 \circ e_2') = \begin{bmatrix} (e_1 \circ e_2) \\ 0
    \end{bmatrix}.\]
    The appended channel of $W'_{\text{line}}$ is zero, so
    \[W'^\top_{\text{line}}(e'_1 \circ e'_2)  + W'^\top_{\text{plane}}e'_{12} + b' = W^\top_{\text{line}}(e_1 \circ e_2)  + W^\top_{\text{plane}}e_{12} + b,\]
    which equals the output of $g \in \mathcal{G}_{k_1}$. Therefore, $\mathcal{G}_{k_1} \subseteq \mathcal{G}_{k_1 + 1}$ and the optimal reconstruction error is nonincreasing in $k_1$.    
\end{proof}

\begin{gaplemma}[GA-Planes satisfies Assumption~\ref{A2} for squared error]
\label{lem:gap-a2}
For all $k_1 \geq 0$,
\[\|x - g_{k_1}^\star\|^2 - \|x - g_{k_1 + 1}^\star\|^2 \geq \|x - g_{k_1+1}^\star\|^2 - \|x - g_{k_1 + 2}^\star\|^2.\]
\end{gaplemma}
\begin{proof}
    By the matrix-completion equivalence in ~\cref{matrix-model}, the squared optimal error decomposes as a joint minimization over a rank-$\leq k_1$ component $UV^\top$ and a low-resolution component $\varphi(L)$. Following~\cite{sivgin2024gaplanes}, for fixed $r_2$, this joint optimization is equivalent to fitting the optimal low-resolution component first and then approximating the residual with a rank-$\leq k_1$ matrix. Let $L^\star$ be the optimal low-resolution fit and define the residual 
    \[R = M - \varphi(L^\star).\]
    Substituting $UV^\top$ as the rank-$k_1$ approximation of $R$ into~\cref{matrix-model} and applying the Eckart-Young theorem yields 
    \[\|x-g_{k_1}^\star\|^2 = \min_{\rank(A) \leq k_1} \| R - A \|_F^2 = \sum_{i > k_1} \sigma_i(R)^2,\]
    where $\sigma_1(R) \geq \sigma_2(R) \geq \cdots \geq 0$ are the singular values of $R$. We can write the consecutive differences as
    \[\|x-g_{k_1}^\star\|^2 - \|x - g_{k_1+1}^\star\|^2 = \sum_{i > k_1} \sigma_i(R)^2 - \sum_{i > k_1 + 1} \sigma_i(R)^2 = \sigma_{k_1 + 1}(R)^2\]
    and
    \[\|x-g_{k_1+1}^\star\|^2 - \|x - g_{k_1+2}^\star\|^2 = \sum_{i > k_1 + 1} \sigma_i(R)^2 - \sum_{i > k_1 + 2} \sigma_i(R)^2 = \sigma_{k_1 + 2}(R)^2.\]
    Since the singular values are nonincreasing, $\sigma_{k_1 + 1}(R) \geq \sigma_{k_1 + 2}(R)$, and so
    \[(\|x-g_{k_1}^\star\|^2 - \|x - g_{k_1+1}^\star\|^2) - (\|x-g_{k_1+1}^\star\|^2 - \|x - g_{k_1+2}^\star\|^2) = \sigma_{k_1 + 1}(R)^2 - \sigma_{k_1 + 2}(R)^2 \geq 0,\]
    which is exactly discrete convexity of the squared optimal error in $k_1$.
\end{proof}

\begin{remark} 
The proof above is stated for the residual low-rank interpretation of 2D GA-Planes: the low-resolution plane component is fixed before analyzing the rank-$k_1$ line component. This is the clean setting in which the Eckart-Young theorem gives an exact expression for the optimal error curve. For a fully joint nonconvex optimization over the low-resolution plane component and the low-rank line component at every $k_1$, the optimal plane component may in principle depend on $k_1$. In that more general setting, the same conclusion will require additional conditions to ensure that the optimal low-resolution component does not change in a way that violates the diminishing-returns behavior.
\end{remark}

\begin{gaplemma}[GA-Planes satisfies Assumption~\ref{A3}]
\label{lem:gap-a3}
There exists a finite $k_1^\star$ such that 
\[\|x - g_{k_1^\star}^\star\| = 0.\]
\end{gaplemma}
\begin{proof}
    Fix $r_1 = N$ so that the linear interpolation at integer coordinates returns the stored grid values exactly. Let $r = \rank(M) \leq N$ and let $M = U\Sigma V^\top$ be the singular value decomposition with $\Sigma = \diag(\sigma_1, \ldots, \sigma_r)$. We exhibit a parameter setting at line feature dimension $k_1^\star = r$ that realizes $M$ exactly. It follows from $r_1=N$ that $e_1[i] = g_1[i,:]$ and $e_2[j]=g_2[j,:]$ for every $(i,j)$. Set
    \[g_1[i,l] = \sqrt{\sigma_l}U[i,l], \qquad g_2[j,l] = \sqrt{\sigma_{l}}V[j,l], \qquad W_\text{line}[l] = 1,\]
    for $l=1,\ldots,r$. Zero out the plane and bias components $g_{12} = 0, W_{\text{plane}}=0, b=0$. Substituting into the model output gives, for every $(i,j)$,
    \begin{align*}
        g_d(\theta) &= W^\top_{\text{line}}(e_1 \circ e_2)  + W^\top_{\text{plane}}e_{12} + b \\
        &= \sum_{l=1}^{k_1}1\cdot g_1[i,l]\cdot g_2[j,l] \\
        &= \sum_{l=1}^{r} \sigma_{l}U[i,l]V[j,l] = M[i,j],  
    \end{align*}
    which is the SVD reconstruction of $M$. Hence, the reconstruction error vanishes at $k_1^\star = r \leq N$ with $r_1 = N$ and any arbitrary $r_2, k_2$. 
\end{proof}

\section{Extensions of the Global Error Bound}
\label{appendix:extensions}
\subsection{Tightness of the $(d^\star - d)$ factor}
\label{appendix:factor_tightness}

The bounds of Section~\ref{sec:theoreticalres} are valid but, as Figures~\ref{fig:task1_error_curves}
and~\ref{fig:task2_error_curves} show, conservative by orders of magnitude. Here, we show that the factor $(d^\star-d)$ is unimprovable under
Assumptions~\ref{A1}--\ref{A3} alone, and that quantifying the rate of
diminishing returns (tightening \ref{A2} to geometric error decrease) removes the $(d^\star - d)$ factor entirely and provides a mathematical foundation for our calibrated error estimates.

\begin{theorem}[Worst-case tightness of the factor $(d^\star-d)$]
\label{thm:factor_tightness}
Consider any finite parameter budget $d < d^\star$. Under Assumptions~\ref{A1}--\ref{A3}, the factor $(d^\star - d)$ in Theorem~\ref{thm:main} is worst-case tight: no uniformly smaller multiplier can be guaranteed from these assumptions alone.
\end{theorem}

\begin{proof}
For tightness, fix any unit vector $u$ and any $\Delta> 0$, and
define $g_j^\star = x-\max\{0,d^\star - j\}\Delta u$ for any $j \geq 0$. Let $e_j := \|x-g_j^\star\|$. Then,
\[
e_j=(d^\star-j)^+\Delta.
\]
It follows that Assumption~\ref{A1} holds and Assumption~\ref{A2} is also satisfied because the marginal error reductions equal $\Delta$ for $j<d^\star$ and zero thereafter. $e_{d^\star}=0$ so Assumption~\ref{A3} also holds. Moreover, for every $j<d^\star$,
\[
\|g_j^\star-g_{j+1}^\star\|=\Delta
\qquad \text{ and } \qquad
\|x-g_j^\star\|=(d^\star-j)\Delta.
\]
Thus, Theorem~\ref{thm:main} holds with equality at every $j<d^\star$; no scaling smaller than $(d^\star -d)$ is valid under Assumptions~\ref{A1}--\ref{A3} alone.
\end{proof}

The constructed instance is exactly the one in which no diminishing returns are available: Assumption~\ref{A2} holds with equality at every increment, so the optimal error decreases linearly in $d$ and the telescoping step of Theorem~\ref{thm:main} loses nothing. However, Theorem~\ref{thm:factor_tightness} does not
preclude improvements for specific architectures whose curves are known to decay faster, and the following theorem gives one such improvement.

\begin{theorem}[A tighter bound under geometric diminishing returns]
\label{thm:geometric} 
Let $e_j := \|x-g_j^\star\|$ and $\delta_j := e_j-e_{j+1}$.
Suppose Assumptions~\ref{A1} and~\ref{A3} hold, and that Assumption~\ref{A2} is strengthened to a geometric decay condition requiring $\delta_{j+1}\leq \rho\,\delta_j$ for $j=d,\ldots,d^\star-2$ and some $\rho\in[0,1)$. Then
\[
\|x-g_d^\star\| \leq
\frac{1-\rho^{d^\star - d}}{1-\rho}
\|g_d^\star-g_{d+1}^\star\|
\leq
\min\left\{d^\star - d,\frac{1}{1-\rho}\right\}
\|g_d^\star-g_{d+1}^\star\|.
\]
\end{theorem}

\begin{proof}
Assumption~\ref{A3} and telescoping give $e_d=\sum_{i=0}^{d^\star - d-1}\delta_{d+i}.$ The geometric-decay condition implies
$\delta_{d+i}\leq\rho^i\delta_d$, so
\[e_d \leq \delta_d\sum_{i=0}^{d^\star - d -1}\rho^i = \frac{1-\rho^{d^\star - d}}{1-\rho}\delta_d.\]
Finally, $\delta_d\leq \|g_d^\star-g_{d+1}^\star\|$ by triangle inequality. The second inequality follows from
$\sum_{i=0}^{d^\star - d -1}\rho^i\leq d^\star - d$ and
$\sum_{i=0}^{d^\star - d -1}\rho^i\leq(1-\rho)^{-1}$.
\end{proof}

\begin{remark}[Interpretation of the calibration constant $c$]
\label{rem:c-interpretation}
Theorem~\ref{thm:geometric} replaces the budget-dependent factor $(d^\star-d)$ with the budget-independent constant $(1-\rho)^{-1}$, which is determined by the architecture's rate of diminishing returns $\rho$ rather than by the signal or the compression level. This offers a meaningful interpretation of the single per-architecture constant $c$ calibrated in
Section~\ref{subsec:constant-c}: $c$ plays the role of an effective $(1-\rho)^{-1}$, and its observed transferability across signals, dimensionalities, and compression levels is exactly what a decay rate fixed by the architecture would predict. We emphasize that the
correspondence is interpretive rather than exact, since $c$ is fit empirically and also absorbs the optimization-gap terms of
Theorem~\ref{thm:main_opt_gap} and the envelope width $U(d)-L(d)$ for architectures where Theorem~\ref{thm:main_envelope} applies.
\end{remark}

\subsection{Relaxing the exact-fit assumption}
\label{appendix:approximateA3}

Assumption~\ref{A3} requires that the parameterization represent $x$ exactly at some finite budget $d^\star$, which can often be large in theory and in practice. For architectures outside the families we study, exact representation may be unavailable or difficult to certify. Since the assumption only enters our proofs to anchor a telescoping sum at a budget where the residual error is known, any budget with a known small residual serves the same purpose. Consider the following relaxation of Assumption~\ref{A3}. 

\paragraph{Assumption A3$'$ (Approximate zero error).}
For some $\epsilon\geq 0$, there exists a finite budget
$d^\star(\epsilon)$ with
\[\|x-g_{d^\star(\epsilon)}^\star\|\leq\epsilon.  \tag{A3$'$} \label{A3'}\]
We take $d^\star(\epsilon) := \min\{\, d \ge 0 : \|x - g^\star_d\| \le \epsilon \,\}$
whenever this set is nonempty. Under Assumption~\ref{A1}, $d^\star(\epsilon)$ is nonincreasing in $\epsilon$, and under
Assumption~\ref{A3} we have $d^\star(0) = d^\star$, so~\ref{A3'} is implied by~\ref{A3} for every $\epsilon \ge 0$ with $d^\star(\epsilon) \le d^\star$.

\begin{theorem}[Global error bound under approximate zero error]
\label{thm:main_relaxed}
Under Assumptions~\ref{A1},~\ref{A2}, and~\ref{A3'}, for any parameter budget $d \leq d^\star(\epsilon)$, the compression error is bounded by
\begin{align}
\|x-g_d^\star\| \leq (d^\star(\epsilon)-d)\|g_d^\star-g_{d+1}^\star\| + \epsilon.
\end{align}
Moreover, for any converged model outputs $\tilde{g}_d, \tilde{g}_{d+1}$ with optimization errors $\eta_d, \eta_{d+1}$, 
\begin{align}
    \|x - \tilde{g}_d\| \leq (d^\star(\epsilon) - d)\|\tilde{g}_d - \tilde{g}_{d+1}\| + (d^\star(\epsilon) - d)\|\eta_{d+1} - \eta_d\| + \|\eta_d\| +\epsilon.
\end{align}
Setting $d^\star(\epsilon)=d^\star$ recovers the
original results in Theorems~\ref{thm:main} and~\ref{thm:main_opt_gap}.
\end{theorem}

\begin{proof}
Telescoping to $d^\star(\epsilon)$ gives
\[\|x-g_d^\star\| = \|x-g_{d^\star(\epsilon)}\|+\sum_{i=1}^{d^\star(\epsilon)-d}(\|x-g^\star_{d+i-1}\| - \|x-g_{d+i}^\star\|).
\]
By Assumption~\ref{A3'}, $\|x-g_{d^\star(\epsilon)}\|\leq\epsilon$, and by Assumption~\ref{A2}, every summand is at most $\|x-g^\star_{d}\| - \|x-g_{d+1}^\star\|$. Hence
\[\|x-g_d^\star\| \leq (d^\star(\epsilon)-d)(\|x-g^\star_{d}\| - \|x-g_{d+1}^\star\|) + \epsilon.
\]
Triangle inequality gives
\[\|x-g^\star_{d}\| - \|x-g_{d+1}^\star\| \leq \|g_d^\star-g_{d+1}^\star\|,\]
which proves the first claim.

For the trained model outputs, triangle inequality yields
\[\|x-\tilde{g}_d\| \leq \|x-g^\star_{d}\| + \|\eta_d\|.\]
In addition,
\[
\|g_d^\star-g_{d+1}^\star\| \leq \|\tilde{g}_d-\tilde{g}_{d+1}\| + \|\eta_{d+1}-\eta_d\|.
\]
Substituting these inequalities into the first bound proves the
second bound.
\end{proof}

Assumption~\ref{A3'} asks only that the model approximately represent $x$ to
accuracy $\epsilon$ at some finite parameter budget, which follows from standard universal approximation or approximation rate statements and does not require the exact interpolation constructions of Lemmas~\ref{lem:grid-a3}, \ref{lem:ffn-a3}, \ref{lem:ingp-a3}, and
\ref{lem:gap-a3}. It is therefore easier to certify for a new architecture
and robust to factors that may make exact representation unachievable in practice, such as finite-precision arithmetic or quantized parameters. Moreover, since $d^\star(\epsilon) \leq d^\star$, this relaxation trades an additive $\epsilon$ against a smaller multiplier, potentially yielding tighter error bounds.

\subsection{Batched Global Error Bound}
\label{appendix:batched}
In practice, estimating $\|g_d^\star - g_{d+1}^\star\|$ from a single adjacent-size pair of models can be sensitive to training instabilities. The following theorem addresses this by averaging the measured gap over a contiguous block of $b$ compression increments, at the cost of $b$ additional model fits.
\begin{theorem}[Batched Global Error Bound]
\label{thm:main_batched}
Under Assumptions \ref{A1} - \ref{A3}, for any parameter budget $d \leq d^\star$ and any block size $1 \leq b \leq d^\star - d$, the compression error is bounded by
\begin{align}
    \|x - g_d^\star\| \leq (d^\star - d)\cdot \frac{1}{b}\sum_{i=1}^{b} \|g_{d+i-1}^\star- g_{d+i}^\star\|.
\end{align}
\end{theorem}
\begin{proof}
We prove the batched version by averaging the one-step decreases of the optimal error curve over the first block of size $b$, and then using discrete convexity to show that every subsequent block has no larger average decrease. Define the one-step improvement
\[
\delta_i := \|x - g_i^\star\| - \|x - g_{i+1}^\star\|, \qquad i = d, d+1, \ldots, d^\star-1.
\]
By Assumption~\ref{A1}, every $\delta_i \geq 0$, and by Assumption~\ref{A2}, the sequence $\delta_d \geq \delta_{d+1} \geq \cdots \geq \delta_{d^\star-1} \geq 0$ is non-increasing. Since $\|x - g_{d^\star}^\star\| = 0$ by Assumption~\ref{A3}, the error telescopes as
\begin{equation}\label{eq:telescope_delta}
\|x - g_d^\star\| = \sum_{i=1}^{d^\star-d} \delta_{d+i-1}.
\end{equation}
We partition the $d^\star-d$ indices into $B = \lceil (d^\star-d)/b \rceil$ consecutive blocks: the first $B-1$ blocks each have size $b$, and the last block has size $b' = d^\star-d-(B-1)b \leq b$. If $B=1$, then the partition consists only of the first block, so the block-sum argument will reduce to the telescoping identity $\|x- g_d^\star\| = \sum_{i=1}^b \delta_{d+i-1}$, and the stated bound follows directly by applying the triangle inequality to each $\delta_{d+i-1}$. We thus consider $B > 1$ and bound each block's contribution separately, using the sum of the first block as a common majorant.

\paragraph{Blocks $k = 1, \ldots, B-1$ (size $b$).}
For the first block $k=1$, the inequality $\sum_{i=1}^{b} \delta_{d+i-1}
  \leq \sum_{i=1}^{b} \delta_{d+i-1}$ holds trivially. For each subsequent block $k \geq 2$, monotonicity of $\delta$ gives $\delta_{d+(k-1)b+i-1} \leq \delta_{d+i-1}$ for each $i = 1,\ldots,b$, since $(k-1)b \geq b \geq 1$. Summing over $i$,
\[
\sum_{i=1}^{b} \delta_{d+(k-1)b+i-1}
  \leq \sum_{i=1}^{b} \delta_{d+i-1}.
\]

\paragraph{Last block (size $b' \leq b$).}
Every index in this block satisfies $d+(B-1)b+i-1 \geq d+b$ for $B \geq 2$, so each summand obeys $\delta_{d+(B-1)b+i-1} \leq \delta_{d+b-1}$. Since $\delta_{d+b-1}$ is the minimum of the first block, it is at most the block average:
\[
\sum_{i=1}^{b'} \delta_{d+(B-1)b+i-1}
  \leq b'\,\delta_{d+b-1}
  \leq b' \cdot \frac{1}{b}\sum_{i=1}^{b} \delta_{d+i-1}.
\]

Substituting these bounds into~\Cref{eq:telescope_delta} yields
\begin{align*}
\|x - g_d^\star\|
  &= \sum_{k=1}^{B-1}\sum_{i=1}^{b} \delta_{d+(k-1)b+i-1}
   + \sum_{i=1}^{b'} \delta_{d+(B-1)b+i-1} \\
  &\leq (B-1)\sum_{i=1}^{b} \delta_{d+i-1}
   + b' \cdot \frac{1}{b}\sum_{i=1}^{b} \delta_{d+i-1} \\
  &= \big[(B-1)b + b'\big] \cdot \frac{1}{b}\sum_{i=1}^{b} \delta_{d+i-1} \\
  &= (d^\star-d) \cdot \frac{1}{b}\sum_{i=1}^{b} \delta_{d+i-1},
\end{align*}
where the last line follows from our construction $(B-1)b + b' = d^\star-d$. Finally, the triangle inequality gives $\delta_{d+i-1} = \|x - g_{d+i-1}^\star\| - \|x - g_{d+i}^\star\| \leq \|g_{d+i-1}^\star - g_{d+i}^\star\|$ for each $i$, so
\[
\|x - g_d^\star\|
  \leq (d^\star-d)\cdot\frac{1}{b}\sum_{i=1}^{b}\|g_{d+i-1}^\star - g_{d+i}^\star\|. \qedhere
\]
\end{proof}

Setting $b=1$ recovers Theorem \ref{thm:main}. Empirically, averaging over a small block of nearby model pairs yields smoother, tighter, and more reliable bound estimates while preserving the same local-to-global interpretation. We therefore use the single-step version as the core theorem and reserve the batched version as a practical refinement requiring training a few additional models at different sizes.

\section{Experimental Details}
\label{appendix:exps}
\subsection{Test Signals}
\label{appendix:test_signals}
We evaluate our framework on a diverse suite of synthetic and real signals spanning 2D and 3D domains, drawn from the benchmark of~\citet{kim2025gridsoutperforminrs} and standard datasets in the computational imaging literature. All signals are either generated at or resized to resolution $128^2$ for 2D and $64^3$ for 3D, ensuring that the model parameter budgets used in our experiments remain in the compressive regime. 

\begin{table}[ht]
\centering
\footnotesize
\begin{minipage}{\linewidth}

\centering
\footnotesize
\caption{\textbf{Fixed hyperparameters for 2D signals.} We use these settings for 2D direct fitting and 2D MRI experiments. For GA-Planes, we use a linear decoder, consistent with the theoretical analysis.}
\vspace{2mm}
\begin{tabular}{|l|l|c|}
\hline
\textbf{Model} & \textbf{Hyperparameter} & \textbf{Value} \\
\hline
\multirow{4}{*}{FFN} 
  & Learning Rate      & 1e-3 \\
  & Embedding Size     & 256 \\
  & Hidden Layers      & 1 \\
  & $\sigma$           & 20 \\
\hline
\multirow{5}{*}{Instant-NGP} 
  & Learning Rate      & 5e-4 \\
  & Hidden Layers      & 1 (2-layer MLP) \\
  & \textit{L}   & 5 \\
  & \textit{F} & 2 \\
  &  Hidden Layer Width  & 64 \\
\hline
\multirow{3}{*}{GA-Planes} 
  & Learning Rate      & 5e-3 \\
  & Line Resolution               & size per axis of signal  \\
  & Plane Resolution   & 2 \\
\hline
\multirow{1}{*}{Grid} 
  & Learning Rate      & 5e-2 \\
\hline
\end{tabular}

\label{tab:hyperparams_2d}

\vspace{3mm}

\caption{\textbf{Fixed hyperparameters for 3D signals.} We use these settings for 3D direct fitting and 3D MRI experiments. For GA-Planes, we use a linear decoder, consistent with the theoretical analysis.}
\vspace{2mm}
\begin{tabular}{|l|l|c|}
\hline
\textbf{Model} & \textbf{Hyperparameter} & \textbf{Value} \\
\hline
\multirow{4}{*}{FFN} 
  & Learning Rate      & 8e-4 \\
  & Embedding Size     & 1520 \\
  & Hidden Layers      & 1 \\
  & $\sigma$           & 5 \\
\hline
\multirow{5}{*}{Instant-NGP} 
  & Learning Rate      & 5e-4 \\
  & Hidden Layers      & 1 (2-layer MLP) \\
  & \textit{L}   & 4 \\
  & \textit{F} & 2 \\
  & Hidden Layer Width & 64 \\
\hline
\multirow{5}{*}{GA-Planes} 
  & Learning Rate      & 8e-3 \\
  & Line Resolution               & size per axis of signal  \\
  & Plane Resolution   & size per axis of signal \\
  & Volume Resolution & 2 \\
  & Volume Feature Dimension & 1 \\
\hline
\multirow{1}{*}{Grid} 
  & Learning Rate      & 5e-2 \\
\hline
\end{tabular}
\label{tab:hyperparams_3d}

\end{minipage}
\end{table}
\FloatBarrier
 
\paragraph{Synthetic signals (2D, 3D).} We construct two families of 2D and 3D synthetic signals parameterized by an explicit notion of effective bandwidth, following~\citet{kim2025gridsoutperforminrs}. \emph{Bandlimited} signals are generated by sampling random noise and applying radial low-pass filters with exponentially spaced cutoffs in the Fourier domain. This yields signals whose frequency content, and thus approximation difficulty, can be smoothly varied from low to high bandwidth. \emph{Sphere} signals consist of randomly placed filled disks (2D) or balls (3D) whose complexity is governed by the number and scale of the primitives: fewer larger objects produce ``low-bandwidth" signals, while many small objects produce fine-grained ``high-bandwidth" structure. We use a fixed intermediate bandwidth for both of these signals for all synthetic-signal experiments.

\paragraph{Real signals.} To complement the controlled synthetic setting, we evaluate on representative real-world data spanning natural images and volumetric signals. For direct signal fitting, we use resized high-resolution images from the DIV2K dataset~\citep{agustsson2017div2k} and 3D Stanford Dragon surface volume~\citep{curless1996dragon}. For inverse problems, we consider two settings with qualitatively different forward models. In MRI reconstruction, we use brain volumes from the ATLAS stroke lesion dataset~\citep{liew2021atlas}, which constitutes an underdetermined linear inverse problem. In radiance field reconstruction, we use the Lego scene from the NeRF-Blender synthetic dataset~\citep{mildenhall2020nerf}, where the forward model is the nonlinear volume rendering integral~\citep{mildenhall2020nerf, max2002optical} and supervision is provided through a photometric loss on posed training images. Together, these datasets introduce broadband spectral content, sharp geometric boundaries, and nontrivial spatial correlations that are not present in synthetic signals but are critical for assessing the performance of our framework beyond idealized settings.

\begin{table}[ht]
\centering
\footnotesize
\begin{minipage}{\linewidth}

\centering
\footnotesize

\caption{\textbf{Fixed hyperparameters for radiance fields.} For radiance field reconstruction, we depart from the shallow decoder theoretical setting by using deeper nonlinear MLP decoders for FFN, Instant-NGP, and GA-Planes. We also use total variation regularization for Grid.}
\vspace{2mm}
\begin{tabular}{|l|l|c|}
\hline
\textbf{Model} & \textbf{Hyperparameter} & \textbf{Value} \\
\hline
\multirow{4}{*}{FFN} 
  & Learning Rate      & 5e-4 \\
  & Embedding Size     & 512 \\
  & Hidden Layers      & 4 \\
  & $\sigma$           & 8 \\
\hline
\multirow{5}{*}{Instant-NGP} 
  & Learning Rate      & 5e-4 \\
  & Hidden Layers      & 2 (3-layer MLP) \\
  & \textit{L}   & 13 \\
  & \textit{F} & 2 \\
  &  Hidden Layer Width & 64 \\
\hline
\multirow{7}{*}{GA-Planes} 
  & Learning Rate      & 5e-4 \\
  & Line Resolution    & 200  \\
  & Plane Resolution   & 16 \\
  & Volume Resolution & 8 \\
  & Volume Feature Dimension & 4 \\
  &  Hidden Layer Width & 128 \\
  & Hidden Layers & 2 \\
\hline
\multirow{2}{*}{Grid} 
  & Learning Rate      & 5e-2 \\
  & Total Variation Weight & 1e-5 \\
\hline
\end{tabular}
\label{tab:hyperparams_nerf}

\vspace{3mm}

\caption{\textbf{Calibrated constants used for predicted error estimates.} For each architecture, the constant $c$ rescales the adjacent-model difference $B(d)$ into the practical predictor $cB(d)$, as described in Section~\ref{subsec:constant-c}. We use a separate constant for radiance fields (``NeRF") because radiance-field reconstruction uses a different rendering pipeline, model-size scale, and view-dependent output structure from the direct fitting and MRI experiments. Within each group of tasks, $c$ is fixed across all datasets, signals, and compression levels.}
\vspace{2mm}
\begin{tabular}{|l|l|c|}
\hline
\textbf{Model} & \textbf{Task} & \textbf{Calibration constant $c$} \\
\hline
\multirow{2}{*}{FFN}
    & NeRF & 3.9 \\
    & All other tasks & 1.4 \\
\hline
\multirow{2}{*}{Instant-NGP}  
    & NeRF & 2.2 \\
    & All other tasks & 0.9 \\
\hline
\multirow{2}{*}{GA-Planes}  
    & NeRF & 4.4 \\
    & All other tasks & 4.6 \\
\hline
\multirow{2}{*}{Grid}  
    & NeRF & 3.1 \\
    & All other tasks & 4.3 \\
\hline
\end{tabular}
\label{tab:c_values}

\end{minipage}
\end{table}
\FloatBarrier

\def\linegap{-0.35mm}
\def\figurewidth{0.993}
\begin{figure}[ht!]
\vskip -3mm
\centering
\includegraphics[width=\figurewidth\linewidth]{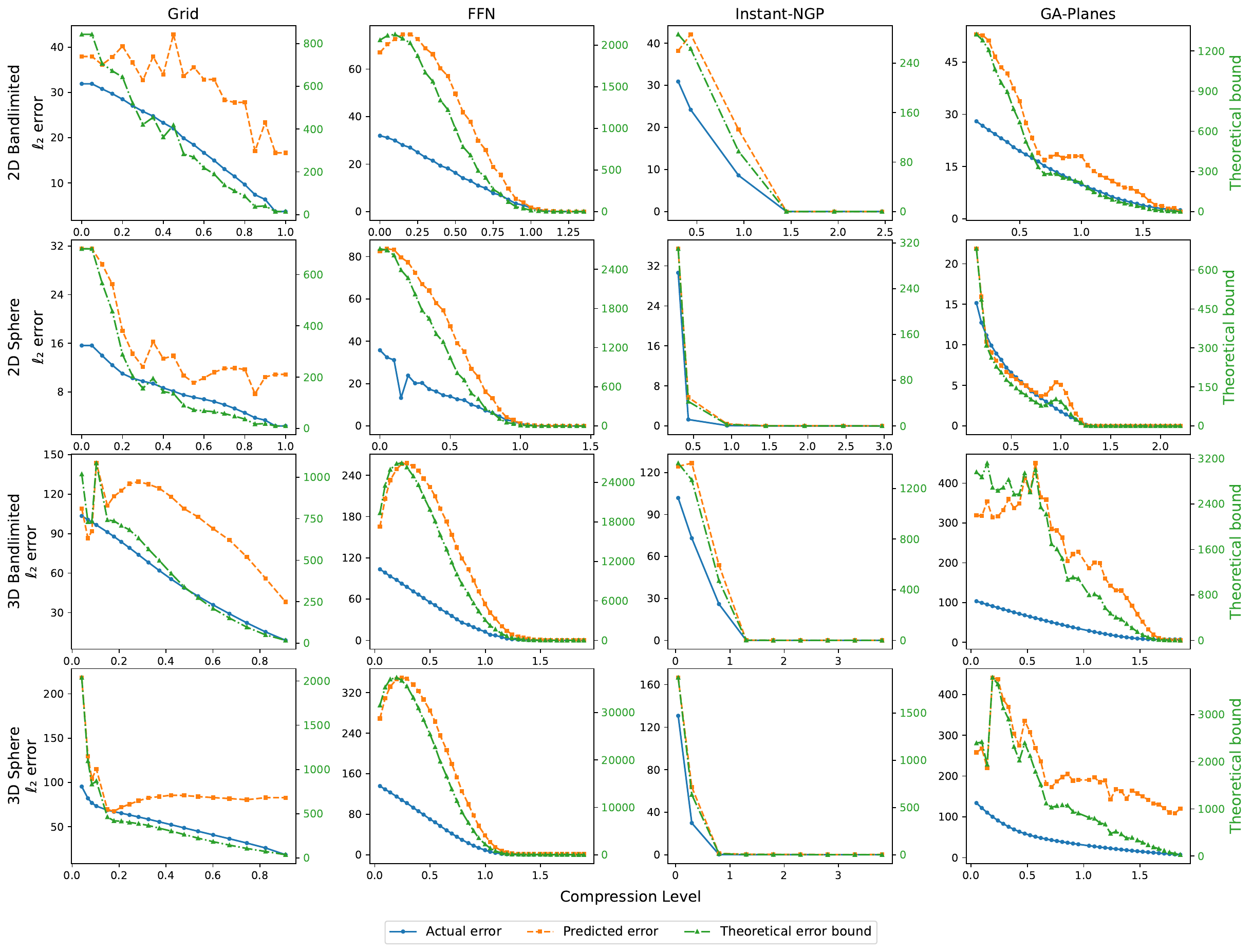}
\vskip 1mm
\begin{minipage}{0.98\textwidth}
    \centering
    \caption{\textbf{Global error curves for direct fitting of synthetic signals.} 
    We fit Grid, FFN, Instant-NGP, and GA-Planes directly to 2D and 3D synthetic signals, including bandlimited and sphere targets. The actual error (blue, left axis) is the $\ell_2$ distance between reconstruction and ground truth, the optimization-gap-adjusted theory-dervied bound (green, right axis) follows from Theorem~\ref{thm:main_opt_gap} as a function of compression level, and the predicted error (orange, left axis) rescales this bound by the constant $c$ applied to the per-increment quantity $B(d)$ (Section~\ref{subsec:constant-c}). Although the theoretical bounds are conservative by orders of magnitude, the calibrated predictions track the actual error across architectures, signal types, and dimensionalities.} 
    \label{fig:task1_error_curves_appendix}
\end{minipage}
\vspace{-6mm}
\end{figure}

\begin{figure}[ht!]
\vskip -3mm
\centering
\includegraphics[width=\figurewidth\linewidth]{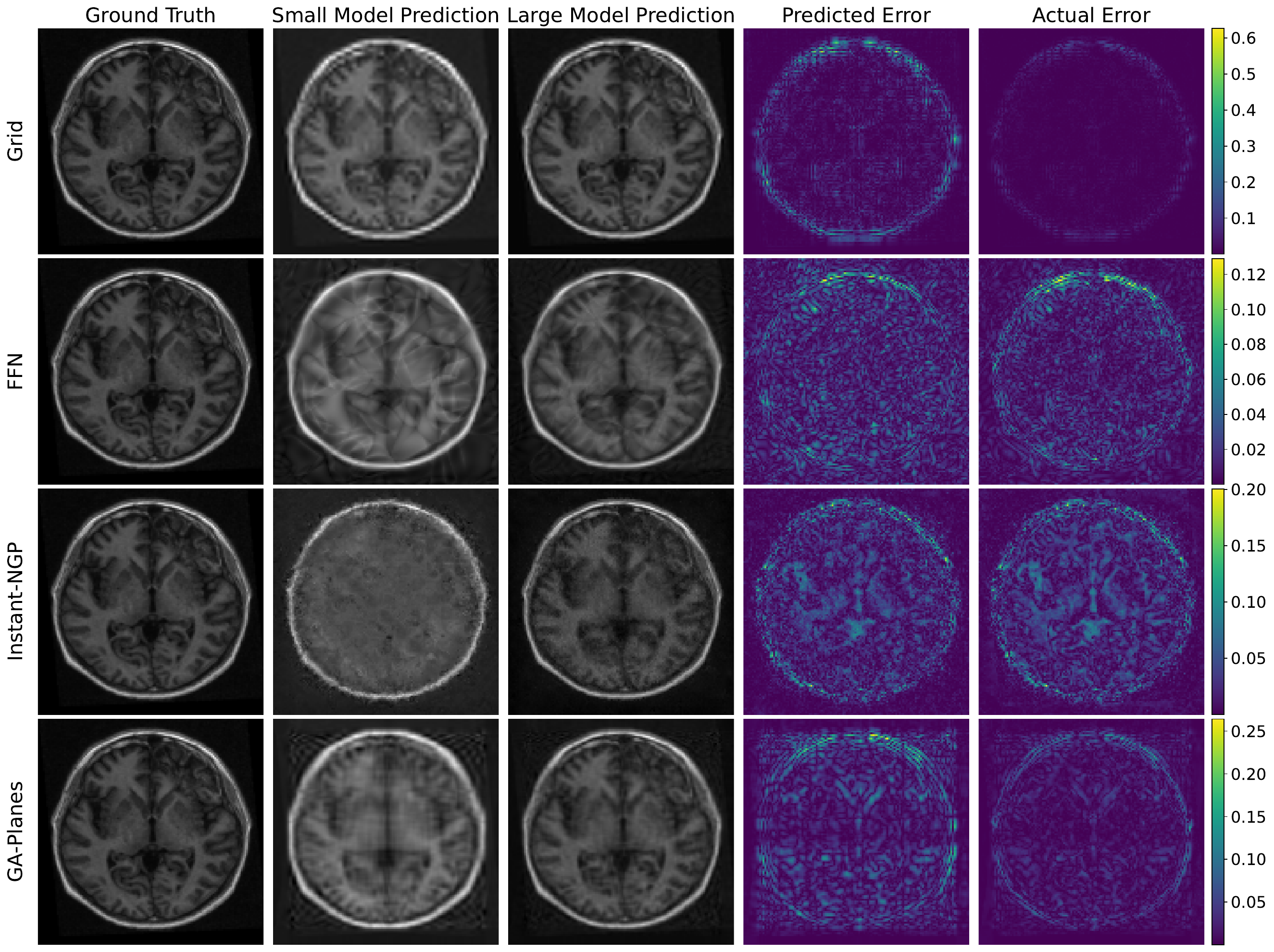}
\vskip 1mm
\begin{minipage}{0.98\textwidth}
    \centering
    \caption{\textbf{Local error heatmaps for 2D MRI slice reconstruction.} 
    Each row corresponds to a different architecture, with small and large models at compression levels $d/n \approx 0.2$ and $0.4$, respectively. Each model maps 2D pixel coordinates to signal values and is supervised indirectly through a $k$-space loss on a single axial slice from the ATLAS dataset~\cite{liew2021atlas}. Across all four representations, the predicted error heatmap $c|\tilde{g}_{\text{small}} - \tilde{g}_{\text{large}}|$ captures the spatial structure of the actual error heatmap $|\tilde{g}_{\text{small}} - x|$, indicating that reconstruction artifacts can be localized without ground-truth access in the inverse-problem setting.}
    \label{fig:2d_mri_heatmaps}
\end{minipage}
\vspace{-5mm}
\end{figure}

\paragraph{Design rationale.} This combination of synthetic and real signals, and direct and indirect training supervision, provides a structured and realistic testbed for validating both the theoretical predictions and empirical behavior of our framework. It ensures that the observed trends, such as the scaling of reconstruction error with model size and the tightness of adjacent-model error predictors, are not artifacts of a specific data regime but instead reflect general properties of differentiable compressive parameterizations.

\begin{figure}[ht!]
\vskip -3mm
\centering
\includegraphics[width=\figurewidth\linewidth]{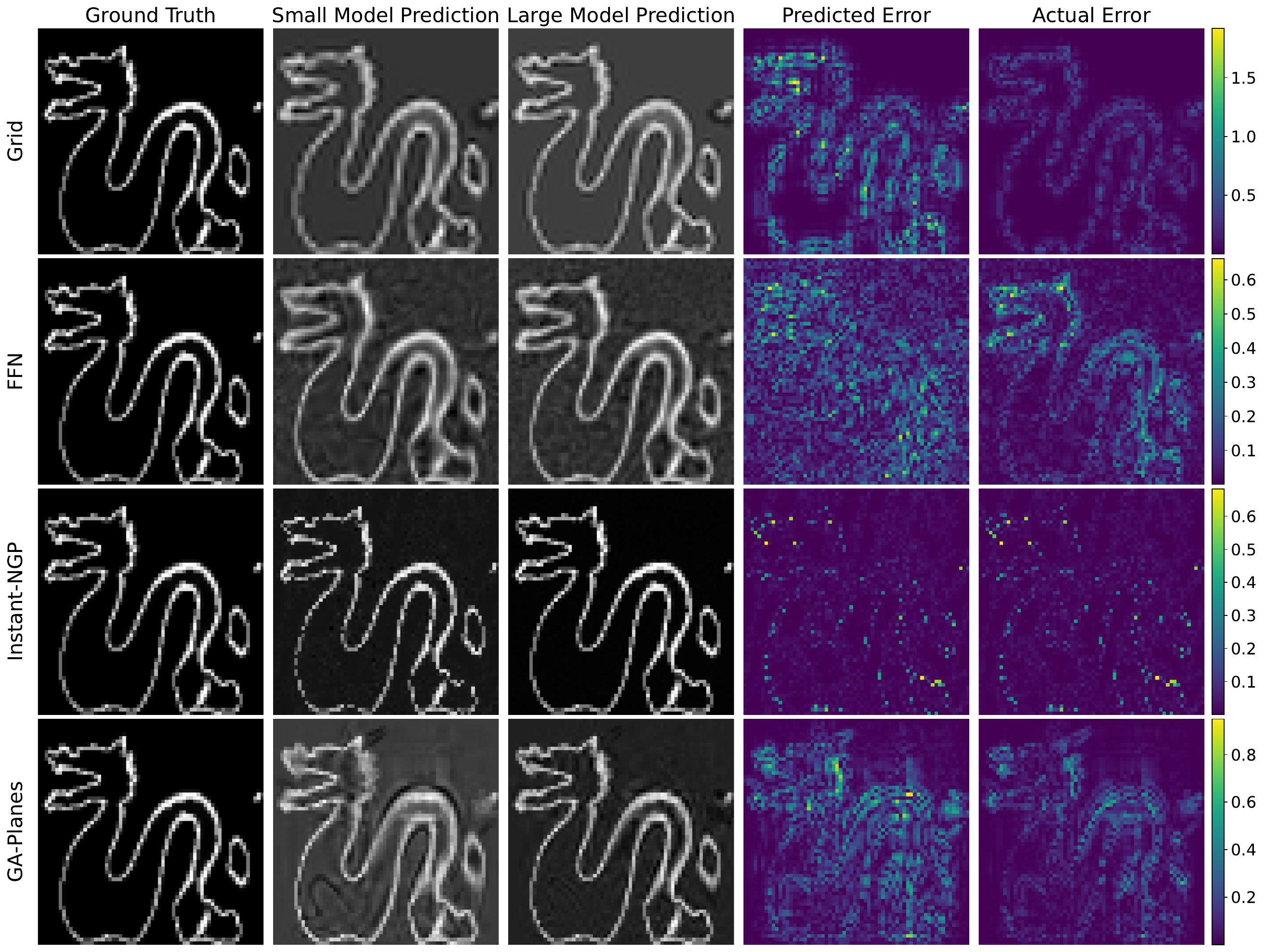}
\vskip 1mm
\begin{minipage}{0.98\textwidth}
    \centering
    \caption{\textbf{Local error heatmaps for direct fitting of 3D Stanford Dragon surface volume~\cite{curless1996dragon}.} Each row corresponds to a different architecture, with small and large models at compression levels $d/n \approx 0.2$ and $0.4$, respectively. Across all four representations, the predicted error heatmap $c|\tilde{g}_{\text{small}} - \tilde{g}_{\text{large}}|$ captures the spatial structure of the actual error heatmap $|\tilde{g}_{\text{small}} - x|$, showing that local compression artifacts can be estimated without ground-truth access for 3D volumes.}
    \label{fig:3D_dragon_heatmaps}
\end{minipage}
\vspace{-5mm}
\end{figure}

\subsection{Implementation Details}
\label{appendix:implementation}
To construct the global error curves, we vary one primary capacity parameter for each architecture while holding all other hyperparameters fixed. Specifically, we vary grid resolution for Grid, MLP width for FFN, hash table size for Instant-NGP, and feature dimension for GA-Planes. The fixed hyperparameters are reported in Table~\ref{tab:hyperparams_2d} for 2D direct fitting and 2D MRI, Table~\ref{tab:hyperparams_3d}  for 3D direct fitting and 3D MRI, and Table~\ref{tab:hyperparams_nerf} for radiance field reconstruction. For 2D and 3D signal fitting and MRI, we use shallow MLP decoders consistent with the theoretical setting. For radiance field reconstruction, we use deeper nonlinear decoders with view-direction inputs, following the standard NeRF setup~\citep{mildenhall2020nerf}, to improve rendering quality in the more challenging view-synthesis setting. All non-radiance-field experiments are implemented in PyTorch while radiance-field experiments are implemented in JAX. 

For direct signal fitting, models are optimized against the observed signal values in image or volume space. For MRI reconstruction, models map spatial coordinates to image intensities and are supervised through a Fourier-domain $k$-space loss. For radiance-field reconstruction, models are trained from posed images using the standard photometric rendering loss. In all experiments, predicted global errors are computed from adjacent trained models using the calibrated architecture-specific constant described in Section~\ref{subsec:constant-c}, and local error heatmaps use the same constant applied pointwise to the difference between small- and large-model reconstructions. For radiance fields, the global error bounds are averaged over 10 test views while the model outputs used to compute the local error heatmaps are rendered from a single held-out test view, matching the evaluation procedure for the actual reconstruction error. We note that since the radiance field experiments use deeper, view-dependent decoders that lie outside our shallow-decoder theoretical setting, we calibrate $c$ separately for radiance fields and for all remaining tasks. Table~\ref{tab:c_values} reports the resulting constants, fit once per architecture per setting following Section~\ref{subsec:constant-c} and held fixed across all signals, datasets, and compression levels. 

All models are trained with the Adam optimizer using $\beta_1 = 0.9$ and $\beta_2 = 0.999$. For optimization gap estimation, we require at least 5 random restarts before considering the restart plateau, and we use a tolerance of $10^{-6}$ for the change in the maximum pairwise reconstruction distance across restarts. Fourier feature networks use Gaussian-initialized Fourier embeddings, following prior work~\citep{tancik2020ffn}. For the radiance field experiments, we additionally apply total variation regularization to the Grid baseline to encourage spatial smoothness and improve reconstruction quality, similar to~\citet{fridovich2022plenoxels}. All experiments were conducted on an NVIDIA RTX A6000 GPU with 48GB VRAM; memory usage did not impose a limiting constraint in our experiments.

\subsection{Global Error Curves: Extended Results}
\label{appendix:exp:global}
We provide additional global error curves for inverse problems, including 2D and 3D MRI reconstruction from k-space measurements and radiance field reconstruction from posed images, in Figure~\ref{fig:task2_error_curves}. Across these tasks, our calibrated predictor tracks the overall decay of the true reconstruction error despite indirect supervision. We also extend the global error evaluation to direct fitting of synthetic signals (from~\citet{kim2025gridsoutperforminrs}) in Figure~\ref{fig:task1_error_curves_appendix}. For each synthetic signal, we vary the model parameter budget while fixing the complexity parameter (effective bandwidth) at 0.3 for sphere signals and 0.9 for bandlimited signals. Across 2D and 3D bandlimited and sphere targets, the raw optimization-gap-adjusted theoretical bound is conservative, but the calibrated adjacent-model predictor follows the decay of the true reconstruction error. These results support the trend observed in the main text: adjacent-model differences provide a stable proxy for global compression error across architectures, signals, and dimensionalities. 

\subsection{Local Error Heatmaps: Extended Results and Quantitative Evaluation}
\label{appendix:exp:local}
We present additional local error heatmaps for two settings: 2D MRI reconstruction from $k$-space measurements (Figure~\ref{fig:2d_mri_heatmaps}) and
direct fitting of the 3D Stanford Dragon surface volume (Figure~\ref{fig:3D_dragon_heatmaps}). Across architectures, the predicted heatmaps recover the main spatial structure of the actual error maps, suggesting that adjacent-model differences remain informative for localizing compression artifacts in both direct-fitting and inverse-problem settings.

We also quantify this consistency by computing the Pearson correlation coefficient between the predicted heatmap $c|\tilde{g}_{\text{small}}-\tilde{g}_{\text{large}}|$ and the actual error heatmap $|\tilde{g}_{\text{small}}-x|$ for each signal in Figures~\ref{fig:2D_div2k_heatmaps}--\ref{fig:3D_nerf_heatmaps} and Figures~\ref{fig:2d_mri_heatmaps}--\ref{fig:3D_dragon_heatmaps}. Since the practically relevant question is not only whether the maps agree but at what spatial resolution, we average-pool both maps into non-overlapping $p \times p$ patches before correlating, sweeping $p \in \{1,2,4,8\}$ with $p=1$ corresponding to the original pixel resolution. Table~\ref{tab:local_heatmap_correlation} reports results for the single signal shown in the corresponding figure, not averages across multiple signals. For every architecture and displayed signal, correlation increases monotonically with patch size. Pixel-level correlations range from 0.624 to 0.782 while $(8\times 8)$-patch correlations range from 0.841 to 0.981. The predicted heatmaps therefore show substantial spatial agreement with the actual errors even at pixel resolution, and stronger agreement at the patch or region-of-interest scale most relevant in practice.

\paragraph{Interpretation and scope.} Since Pearson correlation is invariant to positive scaling, the constant $c$ does not affect these values. The experiment therefore evaluates the localization of the adjacent-model difference itself rather than the calibration of its magnitude. Moreover, these single-example results measure empirical spatial correlation and do not establish a pointwise, region-wise, or distribution-free error guarantee. These results serve to quantify the visual agreement between the predicted and actual error maps across different signals and inverse problems. 
 
\begin{table}[t]
    \caption{\textbf{Pearson correlation between predicted and actual error heatmaps.}
    Both heatmaps are average-pooled over non-overlapping $p\times p$ patches before computing correlation, where $p=1$ is the original pixel resolution. Each block
    corresponds the single signal displayed in the corresponding
    figure; values are not averages over multiple signals.}
    \vspace{2mm}
    \centering
    \setlength{\tabcolsep}{10pt}
    \begin{tabular}{lcccc}
        \toprule
        Model & $p=1$ & $p=2$ & $p=4$ & $p=8$ \\
        \midrule
        \addlinespace[0.4em]
        \multicolumn{5}{l}{\emph{DIV2K direct fitting (\Cref{fig:2D_div2k_heatmaps})}} \\[2pt]
        Grid        & 0.735 & 0.886 & 0.947 & 0.981 \\
        FFN         & 0.624 & 0.767 & 0.881 & 0.974 \\
        Instant-NGP & 0.761 & 0.861 & 0.929 & 0.961 \\
        GA-Planes   & 0.782 & 0.877 & 0.935 & 0.965 \\
        \addlinespace[0.4em]
        \midrule
        \addlinespace[0.4em]
        \multicolumn{5}{l}{\emph{3D MRI reconstruction from k-space(\Cref{fig:3D_mri_heatmaps})}} \\[2pt]
        Grid        & 0.625 & 0.805 & 0.854 & 0.938 \\
        FFN         & 0.666 & 0.719 & 0.776 & 0.841 \\
        Instant-NGP & 0.700 & 0.790 & 0.897 & 0.930 \\
        GA-Planes   & 0.750 & 0.866 & 0.926 & 0.954 \\
        \addlinespace[0.4em]
        \midrule
        \addlinespace[0.4em]
        \multicolumn{5}{l}{\emph{NeRF radiance-field reconstruction (\Cref{fig:3D_nerf_heatmaps})}} \\[2pt]
        Grid        & 0.749 & 0.806 & 0.862 & 0.911 \\
        FFN         & 0.647 & 0.731 & 0.796 & 0.845 \\
        Instant-NGP & 0.781 & 0.869 & 0.919 & 0.949 \\
        GA-Planes   & 0.764 & 0.846 & 0.897 & 0.930 \\
        \addlinespace[0.4em]
        \midrule
        \addlinespace[0.4em]
        \multicolumn{5}{l}{\emph{2D MRI reconstruction from k-space (\Cref{fig:2d_mri_heatmaps})}} \\[2pt]
        Grid        & 0.780 & 0.910 & 0.958 & 0.977 \\
        FFN         & 0.601 & 0.710 & 0.798 & 0.871 \\
        Instant-NGP & 0.870 & 0.927 & 0.960 & 0.975 \\
        GA-Planes   & 0.808 & 0.882 & 0.928 & 0.958 \\
        \addlinespace[0.4em]
        \midrule
        \addlinespace[0.4em]
        \multicolumn{5}{l}{\emph{3D Stanford Dragon surface volume direct fitting (\Cref{fig:3D_dragon_heatmaps})}} \\[2pt]
        Grid        & 0.654 & 0.799 & 0.863 & 0.885 \\
        FFN         & 0.538 & 0.695 & 0.773 & 0.860 \\
        Instant-NGP & 0.988 & 0.992 & 0.996 & 0.998 \\
        GA-Planes   & 0.761 & 0.861 & 0.807 & 0.933 \\
        \addlinespace[0.4em]
        \bottomrule
    \end{tabular}
    \label{tab:local_heatmap_correlation}
\end{table}

\subsection{Model-size Selection without Ground Truth}
\label{appendix:exp:selection}
We evaluate whether our calibrated predictor can support a concrete
downstream decision: given an error tolerance $\tau$, select the
smallest parameter budget whose MRI reconstruction satisfies
\[
    \|x-\tilde{g}_d\|_2\leq\tau,
\]
without observing the test signal $x$. We use an FFN~\citep{tancik2020ffn} model and sweep its capacity over ten compression levels $d/n \in \{0.1, 0.2, \dots, 1.0\}$, using the hyperparameters from Table~\ref{tab:hyperparams_2d}. For each test MRI slice, we fit the image-space FFN indirectly from undersampled
k-space measurements, as in our other 2D MRI experiments. Our predicted error
is $\hat{e}_{\mathrm{adj}}(d)=cB(d)$, using the fixed FFN constant $c$ from Table~\ref{tab:c_values}.

A natural ground-truth-free heuristic for this task is to transfer the error statistics measured on a known dataset. We compute the mean reconstruction error
\[
    \hat{e}_{\mathrm{trans}}(d) = \frac{1}{200} \sum_{i=1}^{200} \|x_i-\tilde{g}_{d,i}\|_2
\]
and its empirical standard deviation $\hat{\sigma}(d)$ for FFN at each compression level over 200 central slices from the Brain subset of FastMRI~\citep{zbontar2018fastmri}. We evaluate two transferred baselines: the mean error $\hat{e}_{\mathrm{trans}}(d)$ and the more conservative mean-plus-standard-deviation estimate
\[
    \hat{e}_{\mathrm{trans+std}}(d) =  \hat{e}_{\mathrm{trans}}(d) +  \hat{\sigma}(d).
\]
Both estimates are fixed across two test signals: a held-out fastMRI brain slice not in the reference set (in-distribution) and an ATLAS~\citep{liew2021atlas} slice (out-of-distribution). The baseline statistical approaches and our method are ground-truth-free at test time but use offline ground-truth calibration: the baselines transfer an entire mean-error curve, whereas our method transfers one architecture-specific scalar $c$ and recomputes the adjacent-model difference using the new measurements. We report results in Table~\ref{tab:mri_selection}, with actual errors included only for retrospective evaluation.

\begin{table}[ht] 
    \caption{\textbf{Ground-truth-free error estimation for 2D MRI.} The transferred baseline is the mean FFN error over 200 central slices from the Brain subset of FastMRI~\citep{zbontar2018fastmri}; ``Ours'' is the calibrated adjacent-model estimate $cB(d)$. Actual errors are shown only for evaluation and are unavailable
    to either selection rule at test time. All entries use the same $\ell_2$ norm convention.}
    \vspace{2mm}
    \centering
    \begin{tabular}{c c cc cc}
    \toprule
    & Transferred FastMRI
    & \multicolumn{2}{c}{Held-out FastMRI}
    & \multicolumn{2}{c}{ATLAS} \\
    \cmidrule(lr){2-2}\cmidrule(lr){3-4}\cmidrule(lr){5-6}
    $d/n$
    & Mean $\pm$ std
    & Actual & Ours
    & Actual & Ours \\
    \midrule
    0.1 & $17.24\pm2.99$ & 17.14 & 35.25 & 18.74 & 39.41 \\
    0.2 & $11.71\pm2.07$ & 14.27 & 29.79 & 17.07 & 35.16 \\
    0.3 & $ 9.07\pm1.72$ & 11.22 & 21.30 & 16.02 & 35.99 \\
    0.4 & $ 8.31\pm1.52$ & 10.15 & 19.07 & 15.03 & 29.57 \\
    0.5 & $ 7.07\pm1.22$ &  8.34 & 14.19 & 14.26 & 26.81 \\
    0.6 & $ 5.76\pm0.99$ &  6.77 & 11.16 & 13.72 & 25.19 \\
    0.7 & $ 4.30\pm0.73$ &  5.06 &  9.95 & 12.09 & 20.39 \\
    0.8 & $ 3.07\pm0.48$ &  3.53 &  7.97 &  9.49 & 19.79 \\
    0.9 & $ 1.39\pm0.22$ &  1.56 &  3.30 &  7.26 & 12.66 \\
    1.0 & $ 0.54\pm0.09$ &  0.56 &  1.57 &  3.87 &  6.77 \\
    \bottomrule
    \end{tabular}
    \label{tab:mri_selection}
\end{table}

Given an estimated error curve $\hat{e}$, we select $\hat{d}_\tau = \min\{d:\hat{e}(d)\leq\tau\}$ over the tested budgets. For $\tau=15$, both baselines select $d/n=0.2$ for both test slices: the mean and mean-plus-standard-deviation error estimates at this budget are $11.71$ and $13.78$, respectively. This choice satisfies the target on the held-out FastMRI slice, whose actual error is $14.27$, but fails on the ATLAS slice, whose actual error is $17.07$. Our calibrated adjacent-model predictor instead selects $d/n=0.5$ and $0.9$, respectively, with corresponding actual errors of $8.34$ and $7.26$, both satisfying the target. Even adding two standard deviations to the transferred mean selects $d/n=0.3$ on ATLAS, whose actual error of $16.02$ still exceeds the target $\tau = 15$. 

This is precisely the failure mode that motivates a signal-adaptive, theory-driven predictor over a transferred statistical estimator. At the same compression level $d/n=0.5$, the transferred baseline assigns both slices the predicted error $7.07$, whereas our estimate increases from $14.19$ on FastMRI to $26.81$ on ATLAS, mirroring the increase in actual error from $8.34$ to $14.26$. Across all ten compression budgets and both test slices, our error estimate remains a valid upper bound in all 20 cases, compared with one case for the transferred mean and five for the mean-plus-standard-deviation baseline. Note that these counts describe empirical coverage on this illustrative two-slice comparison, not a distribution-free guarantee.

\subsection{Calibration Sensitivity and Predictor Tightness}
\label{appendix:calibration-sensitivity}
We quantify the stability of the calibration constant $c$ and the tightness of the resulting predictor over 30 signals: 20 DIV2K images under direct fitting and 10 ATLAS MRI volumes reconstructed from k-space measurements. For each architecture, we evaluate ten compression ratios $d/n \in \{0.1,0.2,\ldots,1.0\}$.

For signal $i$, let
\[ e_i(d) = \|x_i-\tilde{g}_{i,d}\|_2,
    \qquad
    \hat{e}_i(d;c) \coloneq c B_i(d),\]
where $B_i(d)$ is the optimization-gap-adjusted per-increment quantity from
Section~\ref{subsec:constant-c}. We report the prediction ratio
\[R_i(d;c) \coloneq \frac{\hat{e}_i(d;c)}{e_i(d)}, \]
where $R=1$ corresponds to exact error prediction, $R>1$ is conservative but valid, and $R<1$ is an underestimate (invalid error bound). Over the tested compression levels for which $B_i(d)>0$, the smallest constant dominating the observed error on signal $i$ is
\[c_i = \max_d \frac{e_i(d)}{B_i(d)}.\]

\paragraph{Sensitivity of $c$.}
We independently fit $c_i$ on each of the 30 signals and evaluate the resulting predictor on the remaining signals and compression levels. To stress-test the calibration, we define $c_{\mathrm{worst}}:=\min_i c_i$, the least conservative of the 30 fitted values. Table~\ref{tab:calibration-sensitivity} reports the distribution of $c_i$ across these 30 signals, together with the prediction ratio and empirical coverage obtained using $c_{\mathrm{worst}}$. Coverage is the fraction of held-out (signal, capacity) pairs for which $R\geq 1$ (our prediction is a valid upper bound on the true error).

\begin{table}[ht]
    \caption{\textbf{Sensitivity of the calibration constant.} Each $c_i$ is fit on a single signal. $c_{\mathrm{worst}}$ is the least conservative of the $30$ values of $c_i$ across all $30$ signals, i.e. the choice most likely to underestimate compression error. Prediction ratios and coverage are evaluated on held-out (signal, capacity) pairs using $c_{\mathrm{worst}}$. Coverage is empirical and does not constitute a distribution-free guarantee.}
    \vspace{2mm}
    \centering
    \begin{tabular}{lcccc}
        \toprule
        Model
        & Range of $c_i$
        & $c_i$, mean $\pm$ std
        & $R$ under $c_{\mathrm{worst}}$
        & Coverage under $c_{\mathrm{worst}}$\\
        \midrule
        Grid
        & $[4.03,\,4.34]$
        & $4.20 \pm 0.09$
        & $2.08 \pm 0.95$
        & $93.87\%$ \\
        FFN
        & $[1.35,\,1.65]$
        & $1.48 \pm 0.10$
        & $2.32 \pm 0.38$
        & $94.76\%$ \\
        Instant-NGP
        & $[0.63,\,1.12]$
        & $0.81 \pm 0.12$
        & $1.97 \pm 0.55$
        & $95.08\%$ \\
        GA-Planes
        & $[4.52,\,4.83]$
        & $4.63 \pm 0.08$
        & $2.29 \pm 0.36$
        & $94.38\%$ \\
        \bottomrule
    \end{tabular}
    \label{tab:calibration-sensitivity}
\end{table}

The coefficients of variation of $c_i$ are $2.1\%$, $6.8\%$,
$14.8\%$, and $1.7\%$ for Grid, FFN, Instant-NGP, and GA-Planes, respectively. Thus, $c$ is particularly stable for Grid and GA-Planes, moderately variable for FFN, and most signal-dependent for Instant-NGP. Nevertheless, even the least conservative calibration $c_{\mathrm{worst}}$ yields mean prediction ratios of $1.97$--$2.32$ and empirical coverage of the true error on between $93.87\%$ and $95.08\%$ of (signal, capacity) pairs across both image-fitting and MRI reconstruction domains. Sub-$100\%$ coverage signifies exactly the occasional underestimates that the worst-case theoretical bounds of Section~\ref{sec:theoreticalres} exclude, at the cost of orders of magnitude of looseness on average.

\paragraph{Predictor tightness.}
We next use the fixed architecture-specific constants reported in Table~\ref{tab:c_values}, rather than refitting $c$ for individual signals. Table~\ref{tab:predictor-tightness} reports the mean and standard deviation of $R$, aggregated over all (signal, capacity) pairs within each dataset.

\begin{table}[ht]
    \caption{\textbf{Predictor tightness across $30$ signals.} Mean $\pm$ standard deviation of the predicted-to-actual error ratio $R$, using the fixed per-architecture calibration constants from
    Table~\ref{tab:c_values}. Statistics are aggregated over all (signal, capacity) pairs in each dataset. Values above $1$ are conservative.}
    \vspace{2mm}
    \centering
    \begin{tabular}{lcccc}
        \toprule
        Dataset & Grid & FFN & Instant-NGP & GA-Planes \\
        \midrule
        DIV2K
        & $2.54 \pm 1.16$
        & $2.31 \pm 0.65$
        & $2.17 \pm 0.66$
        & $2.26 \pm 0.49$ \\
        ATLAS MRI
        & $1.96 \pm 0.98$
        & $2.40 \pm 0.35$
        & $1.91 \pm 0.62$
        & $2.41 \pm 0.64$ \\
        \bottomrule
    \end{tabular}
    \label{tab:predictor-tightness}
\end{table}

Across four architectures and two datasets, the mean prediction ratio ranges from $1.91$ to $2.54$, so the approximate factor-of-two conservatism observed in Figures~\ref{fig:task1_error_curves} and~\ref{fig:task2_error_curves} reproduces across 30 signals rather than reflecting a favorable choice of example. Accuracy is comparable under the two forms of supervision: no architecture changes its mean ratio by more than $0.6$ between direct fitting on DIV2K and $k$-space supervision on ATLAS, and neither dataset is
systematically easier, with Grid and Instant-NGP predicting slightly tighter bounds on MRI while FFN and GA-Planes yield slightly tighter bounds on images. Grid has the largest variance on both datasets, with standard deviations of $1.16$ on DIV2K and $0.98$ on ATLAS. This is consistent with Grid's Fourier-shell characterization: its approximation error and adjacent-model gap depend directly on the signal's spectral decay, and therefore vary more with signal content than the other architectures.

\paragraph{Scope.}
We note that these results quantify empirical transfer and tightness on the tested signals; they do not convert the calibrated predictor $cB(d)$ into a certified bound. In particular, $R>1$ on average does not imply $R\geq1$ for every (signal, capacity) pair, and the reported coverage is not a population-level or distribution-free guarantee.

\subsection{Computational Overhead}
\label{appendix:computational-overhead}

We measure sequential (fitting one model at a time) wall-clock time on an NVIDIA RTX A6000 GPU, normalized by the time required to train a single model at capacity $d$. We use the same parameter budget sweep of 20 DIV2K images and 10 ATLAS MRI volumes described in Appendix~\ref{appendix:calibration-sensitivity}. All models are trained independently without warm starts; results are shown in \Cref{tab:computational-overhead}.

Our unadjusted adjacent-model predictor trains an additional model at the neighboring capacity $d+1$ and evaluates $c\|\tilde{g}_d-\tilde{g}_{d+1}\|_2$. Because this predictor requires one additional model fit, its total runtime is $1.99\times$--$2.01\times$ that of training a single model; computing the norm between two trained outputs adds negligible time.

Our restart-adjusted predictor additionally estimates the optimization-gap using the restart-based procedure described in Section~\ref{subsec:opt-gap} and evaluates the full optimization-gap-adjusted quantity $cB(d)$.
Estimating the optimization gap is more expensive, requiring $5.09\times$--$6.78\times$ the time of training a single model since it requires independent restarts to estimate optimization gap at each capacity. The larger standard deviation reflects our adaptive stopping rule, under which the restart count is data-dependent rather than fixed. This cost is only incurred when evaluating the optimization-gap-adjusted expression, rather than by the basic adjacent-model calibrated predictor.

\begin{table}[ht]
    \caption{\textbf{Computational overhead.}
    Mean \(\pm\) standard deviation of wall-clock time, normalized by the runtime of one model fit at capacity $d$. The unadjusted predictor includes an independently trained neighboring model, while the
    restart-adjusted predictor additionally includes estimation of the optimization-gap terms $\eta$. We fit all models sequentially in these results.}
    \vspace{2mm}
    \centering
    \begin{tabular}{lcccc}
        \toprule
        Configuration & Grid & FFN & Instant-NGP & GA-Planes \\
        \midrule
        Single fit
        & $1\times$
        & $1\times$
        & $1\times$
        & $1\times$ \\
        Unadjusted predictor
        & $2.01\pm0.14\times$
        & $2.00\pm0.04\times$
        & $1.99\pm0.08\times$
        & $2.01\pm0.08\times$ \\
        Restart-adjusted predictor
        & $5.09\pm1.99\times$
        & $6.78\pm1.64\times$
        & $6.46\pm1.80\times$
        & $6.47\pm1.75\times$ \\
        \bottomrule
    \end{tabular}
    \label{tab:computational-overhead}
\end{table}

These figures are for a single parameter budget evaluated in isolation. In practice, the predictor is often used while sweeping consecutive capacities, which is standard for model selection with implicit neural representations and is the setting of Appendix~\ref{appendix:exp:selection} and Appendix~\ref{appendix:calibration-sensitivity}. In this case, the cost can amortize because each trained model serves as the neighbor of its predecessor: $M$ adjacent-model comparisons require $M+1$ model fits rather than $2M$. Independent fits and restarts are also directly parallelizable on multiple GPUs, which reduces latency but not total compute. The reported measurements are sequential and do not include either form of amortization.

\end{document}